\documentclass[sigconf]{aamas} 

\usepackage{balance} 
\usepackage{amsmath,amsfonts}
\usepackage{algorithmic}
\usepackage{algorithm}
\usepackage{array}
\usepackage{textcomp}
\usepackage{stfloats}
\usepackage{url}
\usepackage{verbatim}
\usepackage{graphicx}
\usepackage{tikz}
\usepackage[capitalise]{cleveref}
\usepackage{amsthm}

\usepackage{caption}
\usepackage{subcaption}

\usepackage{float}
\newfloat{procedure}{htbp}{lop}
\floatname{procedure}{Procedure}

\newtheorem{theorem}{Theorem}[section]
\newtheorem{corollary}{Corollary}[theorem]
\newtheorem{lemma}[theorem]{Lemma}
\newtheorem{proposition}[theorem]{Proposition}
\newtheorem{definition}{Definition}[section]
\newtheorem{claim}[theorem]{Claim}

\newtheorem{conjecture}[theorem]{Conjecture}
\newtheorem*{remark}{Remark}

\newcommand{\abs}[1]{\left\lvert #1 \right\rvert}

\doi{KNSC4490}

\makeatletter
\gdef\@copyrightpermission{
  \begin{minipage}{0.2\columnwidth}
   \href{https://creativecommons.org/licenses/by/4.0/}{\includegraphics[width=0.90\textwidth]{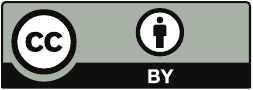}}
  \end{minipage}\hfill
  \begin{minipage}{0.8\columnwidth}
   \href{https://creativecommons.org/licenses/by/4.0/}{This work is licensed under a Creative Commons Attribution International 4.0 License.}
  \end{minipage}
  \vspace{5pt}
}
\makeatother

\setcopyright{ifaamas}
\acmConference[AAMAS '26]{Proc.\@ of the 25th International Conference
on Autonomous Agents and Multiagent Systems (AAMAS 2026)}{May 25 -- 29, 2026}
{Paphos, Cyprus}{C.~Amato, L.~Dennis, V.~Mascardi, J.~Thangarajah (eds.)}
\copyrightyear{2026}
\acmYear{2026}
\acmDOI{}
\acmPrice{}
\acmISBN{}

\title[Robust Sequential Learning in Random Order Networks]{Robust Sequential Learning in Random Order Networks}

\author{William Guo}
\affiliation{
  \institution{University of Pennsylvania}
  \city{Philadelphia, PA}
  \country{USA}}
\email{willguo6@seas.upenn.edu}

\author{Edward Xiong}
\affiliation{
  \institution{MIT}
  \city{Cambridge, MA}
  \country{USA}}
\email{eyxiong@mit.edu}

\author{Jie Gao}
\affiliation{
  \institution{Rutgers University}
  \city{New Brunswick, NJ}
  \country{USA}}
\email{jg1555@rutgers.edu}

\begin{abstract}
In the sequential learning problem, agents in a network attempt to predict a binary ground truth, informed by both a noisy private signal and the predictions of neighboring agents before them.
It is well known that social learning in this setting can be highly fragile: small changes to the action ordering, network topology, or even the strength of the agents' private signals can prevent a network from converging to the truth.
We study networks that achieve \textit{random-order asymptotic truth learning}, in which almost all agents learn the ground truth when the decision ordering is selected uniformly at random. 
We analyze the robustness of these networks, showing that those achieving random-order asymptotic truth learning are resilient to a bounded number of adversarial modifications.
We characterize necessary conditions for such networks to succeed in this setting and introduce several graph constructions that learn through different mechanisms. Finally, we present a randomized polynomial-time algorithm that transforms an arbitrary network into one achieving random-order learning using minimal edge or vertex modifications, with provable approximation guarantees. 
Our results reveal structural properties of networks that achieve random-order learning and provide algorithmic tools for designing robust social networks.
\end{abstract}


\newcommand{\BibTeX}{\rm B\kern-.05em{\sc i\kern-.025em b}\kern-.08em\TeX}

\begin{document}

\pagestyle{fancy}
\fancyhead{}


\maketitle

\section{Introduction}

The ability of a population to correctly aggregate dispersed information is fundamental to effective societal decision-making. Whether predicting a market crash, identifying a public health threat, or creating new social behaviors, success hinges on how weak signals can compound into stronger signals and how these signals propagate through the social network. 
It is in the interest of both individual agents and network designers to understand the circumstances that enable society to learn as a collective, particularly when many network parameters — including the network topology and decision ordering — can be subject to change.  

Social learning, or decision making in a social setting, is a well-studied topic in economics and social networks~\cite{Golub2017-qo,Mobius2014-oy,Bikhchandani1992-rs}.
In the classical model of sequential social learning~\cite{Banerjee1992-ra,Bikhchandani1992-rs,Welch1992-yt,Smith2000-wk,Chamley2004-or}, agents take turns to make decisions, and each agent sees a private signal with uncertainty as well as the previous agent's decisions. One major concern in this setting is \emph{herding} or \textit{information cascades}: with a constant probability, a few early agents make the wrong decision and subsequent agents ignore their own private signals, conforming to the actions of the agents they observe. This results in the majority of the network converging to the wrong value, even with perfectly rational Bayesian agents. 

In this work, we consider sequential learning on a social network, following the model adopted in a few recent papers~\cite{Bahar2020-am,Arieli2021-social,Lu24-enabling,Filip25-NPhard}. A set of agents make predictions sequentially about an underlying binary ground truth. Each agent performs Bayesian inference to compute the truth value with maximum likelihood, using their independent bounded private signal as well as their neighbor's predictions to construct their own belief.
Recent work on this model demonstrates that the quality of information spread throughout a network is dependent on both its graph structure and the order in which individuals act. The ideal learning outcome is characterized by the notion of \textit{asymptotic learning}, in which 
all but $o(n)$ agents in a network of size $n$ correctly learn the ground truth with probability approaching $1$ as $n$ goes to infinity. 
However, the precise relationship between graph topology/decision ordering and whether a network is capable of learning asymptotically still remains elusive.

Of particular interest are networks that achieve asymptotic learning when the decision ordering is chosen \textit{uniformly at random}. 
In practice, imposing a carefully crafted ordering on a network can be challenging without global coordination.
Random orderings result in the ``average case" or expected behavior, allowing us to isolate the role of topology in social learning. 

Due to the inherent difficulty of preserving key graph structures under random orderings, networks that learn under random orderings are more difficult to design. 
In the literature, only two instances of network learning under random orders have been identified: the ``celebrity network'' in \cite{Bahar2020-am}, and symmetric expander graphs satisfying the ``local learning requirement'' in \cite{Arieli2021-social}. 
Both networks seem to possess graph properties robust to adversarial deletion: \cite{Arieli2021-social} explicitly shows the existence of a network obtaining random-order learning, even after the removal of all but an arbitrary constant fraction of the network. 

In this paper, we focus on characterizing and enabling asymptotic learning with random ordering. We characterize the robustness properties of random-order learning networks to adversarial modifications. 
Furthermore, we offer constructions for networks achieving random-order learning through different means, and provide an algorithmic approach to enable arbitrary networks to learn using a small number of modifications.

\subsection{Our Contributions}

To motivate our exposition into random-order learning networks, we first show that strategic-order learning can be quite fragile.
We exhibit a network that achieves asymptotic learning under a strategic ordering only when the uncertainty of the agent's private signals is in a specified range.
This demonstrates the lack of stability in networks that rely on such strategic orderings, as discussed in \cite{Lu24-enabling} and \cite{Filip25-NPhard}. 
In contrast, we show that asymptotic learning under random ordering is robust to a constant number of adversarial modifications, generalizing the robustness properties discussed in \cite{Arieli2021-social} to arbitrary random-order learning networks. Namely, a random-order learning network with learning rate $1 - \varepsilon$ is robust to $o(1/\varepsilon)$ many modifications. This bound is nearly tight in the worst-case, as removing $\Theta(1/\varepsilon)$ many vertices from a learning celebrity network can result in a non-learning network.

Motivated by these robustness properties, we provide constructions for networks that support random-order learning. 
We start with the classic example of a complete network that suffers from an information cascade~\cite{Banerjee1992-ra}.
We first show that any learning network must have an independent set of size $\omega(1)$, a property that a complete graph does not have. 
Surprisingly, by introducing a superconstant number of ``guinea pigs'' (degree one vertices), we find that a complete network can be ``boosted'' to achieve random-order asymptotic learning. This boosting approach is not limited to just guinea pigs: any graph structure capable of achieving asymptotic learning under a strategic order can be embedded to enable the complete graph to enable random-order learning. 

As these constructions are specific to the complete graph, 
one may wonder how to boost an arbitrary network family to achieve random-order learning using a minimum number of edge/vertex modifications. 
We devise a randomized algorithm that achieves this by greedily boosting vertices with high coverage of other non-learning vertices. Our approach is closely related to the influence maximization framework introduced in \cite{Kempe03-maxinfluence}, and more broadly submodular maximization~\cite{Nemhauser1978-jq,Wolsey1982-ik}, which we use to show that our algorithm is a $O(g(n) \log n)$-approximation of the minimum number of edge modifications needed to achieve random-order learning, for any superconstant function $g(n) = \omega(1)$.

Our algorithm, however, assumes the existence of an efficient ``learning oracle'' which determines whether a vertex achieves a high learning rate. The construction of such an oracle appears to be related to many other learning problems that are computationally hard. We conjecture the hardness of computing such an oracle for general graphs. Meanwhile, we give a deterministic local linear-time heuristic to check whether a vertex achieves asymptotic learning.




\subsection{Related Work}
There is a large variation in model and parameter choices within social learning. We focus on the settings closest to ours, and refer the readers to some surveys in the field~\cite{Golub2017-qo,Mossel2017-sd,Mobius2014-oy,Acemoglu2011-tx}. 

\textit{Enabling Social Learning:} To avoid herding in sequential learning, previous work suggests limiting social interactions, for example, removing the visibility between agents early in the sequence~\cite{Smith1991-sy,Sgroi2002-rz}, or introducing a stochastic social network among the agents~\cite{Acemoglu2011-vj}.
A recent work also studied the timing effects of strategic behaviors of agents to delay actions~\cite{Hann-Caruthers2025-sq}.
When the private signals are unbounded (i.e., with probability arbitrarily close to $1$), asymptotic learning occurs almost certainly~\cite{Smith2000-wk, Hann-Caruthers2018-ec}.

\textit{Sequential vs. Repeated learning:} We study sequential learning where each agent makes one prediction following a linear order. In the repeated learning setting, agents broadcast a binary decision at every new time step $t$~\cite{Mossel2014-eb,Mossel_2013}.
For Bayesian agents with repeated learning, asymptotic truth learning is achieved if the graph is undirected and the distribution of private beliefs is not atomic\footnote{Note that a binary Bernoulli distribution for private
signal as in our setting is not non-atomic.}~\cite{Mossel2014-eb}, or if agents repeatedly share with others their current beliefs~\cite{Mossel2017-sd}. 

\textit{Hardness of Bayesian Learning:}
It is well known that general Bayesian belief inference is difficult \cite{Cooper1990-fn}. In the repeated social learning setting, it has been shown that computing the optimal decision is NP-hard even at $t = 2$ \cite{Hazla2021-vf}. Generally, inferring the most likely ground truth value is PSPACE-hard \cite{Hazla2019-reasoning}. In the sequential setting, determining whether a particular network can obtain a high learning rate under some strategic ordering is also NP-hard~\cite{Filip25-NPhard}.

\textit{Influence Maximization:} 
Our algorithm for enabling random-order learning relies on identifying a set of highly influential agents, a problem dating back to \cite{Domingos2001-min}.
Prior work has established approximation algorithms for maximizing the expected influence under various diffusion models ~\cite{Kempe03-maxinfluence, Borgs2014-maxinfl}. 
Notably, so long as the influence objective is submodular, one can immediately obtain strong approximation guarantees~\cite{Wolsey1982-ik, Nemhauser1978-jq}.

\textit{Robustness of Social Learning:}
The existing literature predominantly analyzes the robustness of social learning in light of faulty or misguided decision-making rules. For example, Bohren~\cite{Aislinn_Bohren2016-ya} considers a setting where some agents blindly predict based on their private signal and shows that significantly over-/underestimating the proportion of uninformed agents leads to undesirable behavior. Mueller-Frank~\cite{Mueller-Frank2018-lh} concerns the repeated learning setting, noting that the standard weighted average (DeGroot) decision model is not robust to arbitrarily small adjustments made to a single agent at each iteration. 
In the sequential learning setting, Arieli et. al~\cite{Arieli2021-social} observed that networks learning under random orderings are robust to random vertex deletions, showing that robustness against random vertex deletions is embedded in the definition of uniform random orderings. 

\section{Setting}\label{section:setting}

We define a \textit{network} $\mathcal{F}$ to be an infinite family of undirected graphs with unbounded order $\sup_{G \in \mathcal F} \abs{V(G)} = \infty$. For ease of notation, throughout this paper we treat networks as sequences $\mathcal F = \{ G_n \}_{n \in \mathbb N}$ of undirected graphs $G_n = (V_n, E_n)$ where $\abs{V_n} = n$. Each graph $G_n$ represents $n$ agents with perfect knowledge of the entire graph topology. Each graph is equipped with a decision ordering (equivalently, permutation) $\sigma_n: [n] \to [n]$, where $\sigma_n$ is a bijection mapping each agent to a unique index. We distinguish between the \textit{strategic ordering} setting, in which one may carefully choose the agent's ordering for any graph $G_n$, and the \textit{random ordering} setting, in which $\sigma_n$ is chosen uniformly at random from the $n!$ possibilities.
In either case, the decision ordering is known by all agents in advance. 

On a graph $G = (V,E)$, each agent $v \in V$ seeks to learn the ground truth signal $\theta \in \{0, 1\}$, where $\Pr(\theta = 0) = \Pr(\theta = 1) = \frac{1}{2}$. Each agent $v$ is given a random private signal $p_v \in \{0, 1\}$, representing an independent noisy measurement of the ground truth. The agent's private signals are correlated with $\theta$ with a common probability $q > 1/2$. 

As per the decision ordering $\sigma_n$, agents will go in sequence to make their predictions. For any agent $v$, let $N(v)$ denote the set of neighbors of $v$ in a graph. In sequential learning, $v$
can see the actions of their neighbors arriving before them in the ordering, namely in the set $N_{\sigma}(v) = \{u \in V_n: (u, v) \in E_n, \sigma(u) < \sigma(v)\}$. Each agent $v$ then announces a prediction $a_v \in \{0, 1\}$ visible only to its neighbors succeeding it in the ordering. The ordering $\sigma$ induces an acyclic directed graph $\vec G$ which orients each edge in $G$ to point from the lower index to the higher index in $\sigma$. 
An agent $v$'s \textit{subnetwork} $B(v)$ is defined to be the set of all $u \in \vec{G}$ for which a $u$-$v$ directed path exists in $\vec G$. Such vertices are called the ancestors of $v$. Informally, it is the set of all agents that may ultimately influence $v$'s belief and action.

Following the literature in social learning~\cite{Banerjee1992-ra}, we assume all agents are perfectly Bayesian, and will predict the most probable ground truth value given their knowledge of $G, \sigma, q,$ their private noisy signal, and the predictions of their neighbors arriving before them.
For simplicity, if an agent finds that it is equally likely for $\theta = 0$ or $1$, the agent will side with their own private signal, although our results generalize to arbitrary tie-breaking rules.

For each agent $v \in V$, its \textit{learning rate} $\ell_{\sigma}(v) := \ell_{\sigma}(v, q)$ given a fixed ordering $\sigma$ is the probability that $v$ correctly predicts the ground truth: $\ell_{\sigma} (v) = \Pr(a_v = \theta)$.
The learning rate of a graph $G = \{V, E\}$ paired with a fixed ordering $\sigma$ is the expected fraction of agents predicting the ground truth correctly, i.e.
$$L_{\sigma}(G) = \frac{1}{|V|} \sum_{v \in V} \ell_{\sigma}(v).$$

Similarly, we define the \textit{random-order learning rate} of an agent $v$ to be $\ell(v) = \mathbb{E}_{\sigma}[\ell_\sigma(v)]$, averaging the learning rate over all possible orderings $\sigma$ weighted uniformly. Denote by $S_n$ the set of all possible permutations of $[n]$. For an event $A \subseteq S_n$ dependent only on the random ordering of the vertices, the random-order learning rate conditioned on $A$ is denoted as $\ell(v \mid A) = \Pr(a_v = \theta \mid \sigma \in A)$.
The random-order learning rate of a graph $G$ is the expected fraction of agents predicting the ground truth given a decision ordering selected uniformly at random, or
$$L(G) = \mathbb{E}_{\sigma}[L_{\sigma}(G)] = \frac{1}{|V|}\sum_{v \in V} \ell(v).$$

The learning rate of a network $\mathcal F = \{G_n\}_{n \in \mathbb{N}}$ is the \textit{asymptotic learning rate} of its graph family, or $\lim_{n \to \infty} L(G_n)$.
We say that a network obtains \textit{random-order asymptotic truth learning} if $L(G_n) \to 1$ as $n \to \infty$.  
When $\mathcal{F}$ is equipped with a sequence of fixed orderings $\{\sigma_n\}$, we say the network achieves \textit{strategic order learning} if $\lim_{n \to \infty} L_{\sigma_n}(G_n) = 1$. We will abbreviate random order (resp. strategic order) asymptotic truth learning as random-order (resp. strategic order) learning, or even simply ``learning'' when the context is clear. An agent $v$ with learning rate approaching $1$ as $n$ goes to infinity \textit{learns} or \textit{achieves learning}.

We now describe a few helpful results with proofs in \cref{appendix:background}. The first is the generalized improvement principle, an extension of Proposition $1$ in \cite{Lu24-enabling}:

\begin{proposition}[Generalized improvement principle]
\label{prop:improvement2}
    Given a graph $G = \{V, E\}$ and a particular vertex $v \in V$, consider a  subgraph $G' = \{V, E'\}$ of $G$ such that $E' \subseteq E$, and for all $u, w \neq v$, $(u, w) \in E \implies (u, w) \in E'$. 
    Then for any fixed ordering $\sigma$, the learning rate of $v$ in graph $G$, denoted by
    $\ell_{\sigma}(v)$, is at least that of the learning rate of $v$ in $G'$ denoted by $\ell'_{\sigma}(v)$. Further, under random orderings $\ell(v) \geq \ell'(v)$. 
\end{proposition}

The generalized improvement principle ensures that each agent does at least as well as if it could only see a subset of its neighbors.
This is true due to the rationality (Bayesian inference) of each agent, along with their knowledge of the graph $G$ and ordering $\sigma$. 
As a corollary, each agent learns at least as well as its neighbors before it in the decision ordering; inductively, this implies that each agent does at least as well as all previous agents in its subnetwork.

Next, we establish a general property of random-order learning in a network, which holds independently of the decision rules.

\begin{lemma}
    \label{lemma:conditionallearning}
    For a vertex $v$ in $G$ with a random learning rate at least $\ell(v) \geq 1 - \varepsilon$, and an event $A \subseteq S_n$ which occurs with probability $\Pr(A) = \frac{|A|}{n!}$ on uniform random orderings, the learning rate $ \ell(v \mid A)$ of $v$ conditioned on $A$ is bounded by
    \[
    \ell(v \mid A) \geq 1 - \frac{\varepsilon}{\Pr(A)}.
    \]
\end{lemma}

Finally we note that any vertex with many learning neighbors itself must achieve learning.

\begin{lemma}
\label{cor:goodneighborhood}
    For any network $\mathcal{F} = \{G_n\}_{n \in \mathbb N}$ and vertex $v = \{v_n : v_n \in G_n, n \in \mathbb{N}\}$, if $v$ has $d = \omega(1)$ neighbors that achieve learning, then $v$ itself also achieves learning.
\end{lemma}

One may expect that as long as the random ordering places $v$ after one of its $d$ learning neighbors, it too will learn. However, each neighbor has a lower learning rate conditioned on being early in the ordering, as, in expectation, they have less information to inform their beliefs (see \cref{lemma:monotonicity}).
If $v$ arrives too early, none of its neighbors before it may be in a position that enables learning.
The decline in learning rate conditioned on being early in an ordering is a recurring obstacle throughout our analyses.

\begin{remark}
    Many of our results invoke asymptotic notation despite referencing a single instance of a graph or vertex. As these results generalize to an arbitrary $n$, we find it acceptable to do this for brevity, as is standard for theoretical findings.  When discussing guarantees related to `asymptotic learning' for single graph/vertex instances, as in \cref{cor:goodneighborhood}, we will occasionally clarify our results by referencing infinite families/sequences as opposed to fixed instances.
\end{remark}

\section{Robustness in Network Learning}

\subsection{Fragility of Strategic-Order Learning}\label{subsection:fragile}

To motivate our exposition into random-order learning networks, we first highlight the fragility of social learning under strategic orderings. Prior work (e.g., \cite{Lu24-enabling}) studied a myriad of networks that learn under strategic orderings. 
This type of learning is heavily dependent on the strategic ordering, yet even so, it can be fragile to other network parameters. 
Here, we introduce a network family that learns conditionally depending on $q$ (the probability that private signals are correct), even with a fixed graph topology and strategic ordering of the vertices. 

Our construction is shown in Figure \ref{fig:smarter-dumber}, in which there is a special vertex $v$ with four neighbors arriving before it, and a set of $k$ neighbors $\{w_i\}$ arriving after it. Each $w_i$ also has an independent neighbor that arrives before it. All the agents $\{w_i\}$ feed their predictions into a common neighbor $u_0$ arriving later. We let $k=\log n$, though any superconstant value of $k$ suffices. All remaining vertices are arranged in a chain following the vertex $u_0$.

In essence, the ability of the network to learn depends solely on whether $u_0$ is a learner. If $u_0$ could observe all (or even a constant fraction) of the independent private signals from the agents $\{w_i\}$, it would be able to correctly determine the ground truth with high probability by a simple Chernoff bound. However, all agents $w_i$ observe the action of $v$. For $q > q_0$ for some $q_0 \approx 0.7887$, each $w_i$ will default to the action of $v$, whereas for $q < q_0$, each  $w_i$ will act on its independent private signal with a constant probability. In the former case, all $w_i$'s will simply regurgitate the action of $v$, which $u_0$ will also adopt. But as $v$ is not a learner, neither is $u_0$. In the latter case, $u_0$ receives enough private signals to learn with high probability. The detailed proof can be found in \cref{appendix:fragile}.

\begin{proposition}\label{prop:changeq}
    There exists a strategic-order network $\mathcal F$ that obtains asymptotic learning for $q < q_0$ and does not learn for $q > q_0$, for some fixed threshold $q_0$.
\end{proposition}

\begin{figure}
    \begin{tikzpicture}
        \begin{scope}[every node/.style={circle, draw}]
            \node (v) at (2.5, 1.5) {$v$};
            \foreach \i in {1,...,4} {
                \node (v\i) at (\i, 0) {};
                \path [->] (v\i) edge (v);
            }
            \node (w) at (2.5, 4.5) {$u_0$};
            \foreach \i in {1,...,4} {
                \node (w\i) at (4.5, \i + 1) {$w_\i$};
                \node (ww\i) at (6, \i + 1) {}; 
                \path [->] (ww\i) edge (w\i);
                \path [->] (v) edge (w\i);
                \path [->] (w\i) edge (w);
            }
        \node (wnext) at (1, 4.5) {};
        \path [->] (w) edge (wnext);
        \end{scope}
        \node (vdots) at (4.5,6) {$\vdots$};
        \node (dots) at (-0.5, 4.5) {$\cdots$};
        \path [->] (wnext) edge (dots);
    \end{tikzpicture}
    \caption{A network that obtains learning for $q < q_0$ but not for $q > q_0$, for some $q_0 \in [\frac{1}{2}, 1]$. The edges in the graph are oriented to indicate the direction of information flow.}
    \label{fig:smarter-dumber}
\end{figure}
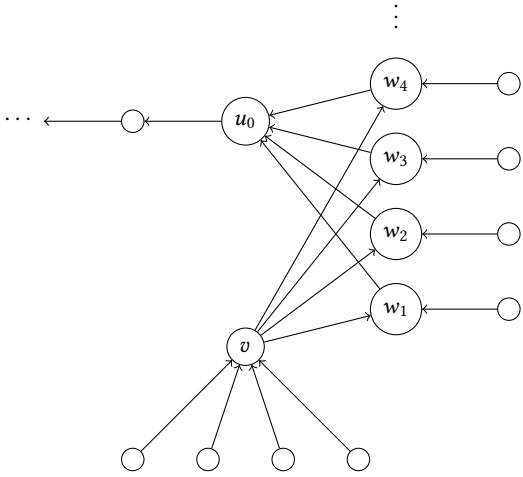

This example is rather counter-intuitive in that the network stops learning when the quality of private signals \emph{improves}.
However, this is not universal: a similar construction of a network that only learns for \textit{high} values of $q$ can be found in \cref{appendix:fragile}. It can be seen that a small modification to the network, e.g., inserting or removing one of the early-arriving neighbors of $v$, can influence the entire network's learning behavior via a similar analysis.

A fundamental strength of random-order learning networks is the ability to learn almost independently of the decision ordering imposed on it. Such networks are able to learn seemingly due to inherent properties of the underlying topology. 
A natural question is whether such properties can be broken via minor perturbations to the network structure.
In the next section, we investigate the robustness of random-order learning networks against \textit{adversarial} edge/vertex modifications. 


\subsection{Robustness of Random-Order Learning}\label{subsection:robust}

In this section, we argue that a network's robustness to edge/vertex insertion/deletions can be made explicit in terms of its random-order learning rate. We start with a rather intuitive observation: conditioned on being early in the decision ordering, an agent does worse than when conditioned on being later in the ordering. This can be attributed to the generalized improvement principle, as an agent arriving early and only seeing a few of its neighbors will be less informed than an agent arriving later and seeing more of its neighbors.

\begin{lemma}[Monotonicity]
\label{lemma:monotonicity}
        Given a graph $G$, fix an agent $v \in G$. For all $i \in [n]$, let $A_i$ be the event that $v$ is at index $i$, or $\sigma(v) = i$. Then $\Pr(a_v = \theta \mid A_i) \leq \Pr(a_v = \theta \mid A_{i+1})$ for all $1 \leq i \leq n - 1$. 
\end{lemma}

A simple proof can be found in \cref{appendix:background}.
Using monotonicity, we demonstrate that the random-order learning error of a graph only increases linearly in the number of modifications.
The proof relies on the probability of any given agent arriving earlier in the random ordering before any of the vertex modifications. 
    
\begin{theorem}
    Given a graph $G$ and a vertex $v \in G$ with random-order learning rate $1 - \varepsilon$, after $k$ vertex deletions, $v$ obtains a learning rate at least $1 - (k+1)\varepsilon$.
\end{theorem}

\begin{proof}
    Let $D$ be the set of $k$ deleted vertices, and let $G'$ be the graph obtained by deleting $D$ from $G$. For $n = \abs{G}$, the size of $G'$ is $n - k$. 
    Consider the probability distribution $\mathbb P_{v}$ of orderings on $G'$, obtained by first uniformly at random picking an ordering $\sigma$ of $G$ such that $v$ precedes all $w \in D$, and then removing all the $w \in D$ to obtain an ordering $\sigma'$ on $G'$.    We use $\ell(v \mid \mathbb P_{v})$ to denote the learning rate of $v \in G'$ under this distribution of orderings. Furthermore, let $A_v$ denote the event that $v$ precedes all vertices $w \in D$ in the original graph $G$.
    
    The probability $\Pr(A_v)$ under a uniformly random ordering of vertices in $G$ is $\frac{1}{k + 1}$, so by Lemma \ref{lemma:conditionallearning} we have
    \[
    \ell(v \mid A_v) \geq 1 - (k +1)\varepsilon.
    \]
    For any ordering $\sigma$ such that $A_v$ occurs, we can trivially remove all vertices in $D$ to obtain an ordering $\sigma'$ of $G'$. The learning rate of $v$ under $\sigma$ and $\sigma'$ is the same, since $v$ is not affected by any vertex that arrives after it. Therefore, 
    \[
    \ell(v \mid \mathbb P_v) = \ell(v \mid A_v) \geq 1 - (k + 1) \varepsilon.
    \]
    Now observe that in the distribution $\mathbb P_v$, the probability that $v$ is found at index $i$ is
    \[
    \Pr(\sigma(i) = v \mid \mathbb P_v) = \frac{\binom {n-i}{k}}{\binom {n}{k+1}}
    \]
    and for the orderings such that $\sigma(i) = v$, all other vertices are uniformly random in the remaining indices.
    Importantly, note that $\Pr(\sigma(i) = v \mid \mathbb P_v)$ is monotonically decreasing in $i$. Thus by Lemma \ref{lemma:monotonicity}, the learning rate of $v$ in $G'$ is
    \begin{align*}
        \ell(v) &= \frac{1}{n-k}\sum_{i=1}^{n-k} \ell(v \mid \sigma(i)=v)
        \geq \sum_{i=1}^{n-k} \frac{\binom{n-i}{k}}{\binom{n}{k+1}} \ell(v \mid \sigma(i)=v) \\
        &= \ell(v \mid \mathbb P_v) 
        \geq 1 - (k+1) \varepsilon
    \end{align*}
    as desired.
\end{proof}

The same approach works to prove the following analogous statements on other modifications. 

\begin{corollary}
    Given a graph $G$ and a vertex $v \in G$ with random-order learning rate $1 - \varepsilon$, after $k$ vertex insertions, $v$ obtains learning rate at least $1 - (k+1)\varepsilon$.
\end{corollary}

\begin{corollary}
    Given a graph $G$ and a vertex $v \in G$ with random-order learning rate $1 - \varepsilon$, after $k$ edge deletions or insertions, $v$ obtains learning rate at least $1 - \frac{2k+1}{2}\varepsilon$.
\end{corollary}

In this context, we note that a vertex insertion allows for the addition of an arbitrary set of edges from the new vertex to any existing vertices in the graph. Thus, we find that networks achieving random asymptotic learning are robust against a number of modifications dependent on the learning rate of the network:

\begin{theorem}\label{thm:randrobust}
    Given a network $\mathcal F$ with random-order asymptotic learning rate $1 - \varepsilon$, define the network $\mathcal F'$ where each $G'_n \in \mathcal F'$ is obtained from some $G_m \in \mathcal F$ after at most $k$ modifications, which may be vertex insertions/deletions or edge insertions/deletions. Then, $\mathcal F'$ obtains random-order learning rate at least $1 - (k+1)\varepsilon$. If $\varepsilon = o(1)$ and $k = o(1/\varepsilon)$, then $\mathcal{F'}$ achieves asymptotic learning. 
\end{theorem}

For general random-order learning networks, this bound is nearly tight.
As an example, consider the celebrity network \cite{Bahar2020-am}, which consists of complete bipartite graphs $K_{n-k, k}$ with $k = o(n)$ and $k = \omega(1)$.\footnote{When $k = \Omega(n)$ or $O(1)$, a poorly informed herd occurs with a constant probability. For example, when $k = \Omega(n)$, with a constant probability, the first celebrity sees only a constant number of commoners, so its probability of failure is at least a constant.} The larger partition of $n-k$ vertices can be thought of as ``commoners", while the smaller partition of $k$ vertices represents ``celebrities". In a given random ordering, with high probability, the first celebrity in the ordering observes the independent signals of many commoners. Upon aggregating this information into a single high-quality signal, the celebrity disseminates it to the rest of the network. When $k = \Omega(\log n)$, the learning rate of the celebrity network is at most $1 - \Theta(1/k)$, as at least $\Theta(n/k)$ commoners arrive before the first celebrity with probability $1 - o(1)$. Thus for $\varepsilon = \Theta(1/k)$, one may adversarially delete the $k = \Omega(1/\varepsilon)$ celebrities from the network, leaving $n - k$ lone vertices which do not learn. 

\begin{figure}
    \centering
    \includegraphics[width=0.6\linewidth]{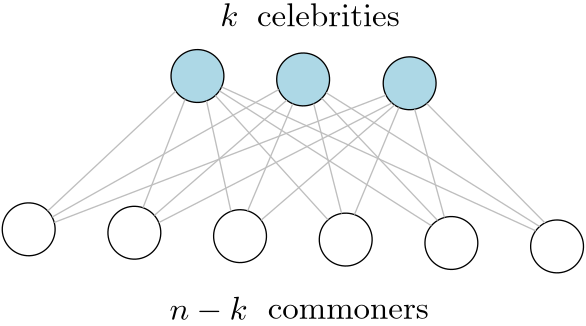}
    \caption{The celebrity graph as described in \cite{Bahar2020-am}} 
    \label{fig:celebrity}
\end{figure}

\section{Constructing Robust Networks}\label{section:algorithm}

\subsection{Boosting a Complete Network}\label{subsection:randomlearningexample}

When given a network that does not support random-order learning, a natural question is to ask how to achieve learning with minimal modifications, especially given the robustness guarantees proven in \cref{subsection:robust}.
Prior to this work, the understanding of networks obtaining random-order learning has been quite limited, both in terms of necessary/sufficient properties and tangible examples of graph families. We first note a key necessary property of random-order learning networks, relying on a simple observation of Bayesian learning: the strength of a signal is dependent on the number of independent private signals backing it. It is known that each decision is a deterministic function of the announced private signals before it (see Lemma 4 in Appendix B of \cite{Lu24-enabling}). 
Then, for any agent $v$ learning under random orderings, it is necessary that, with probability approaching $1$, a superconstant number of agents preceding $v$ act according to their own private signals. Conversely, if there is a constant probability that $v$'s subnetwork contains only $O(1)$ ancestors acting on their own private signals, with constant probability all of those ancestors receive incorrect private signals, implying that $v$ has a learning rate bounded away from $1$. 
We relate the number of agents predicting their own private signal to the size of the maximum independent set, and formalize a network-wide property in \cref{prop:indep_set}. The proof can be found in \cref{appendix:results4}.

\begin{proposition}\label{prop:indep_set}
    Any network $\mathcal{F}$ obtaining random-order learning must have maximum independent set size $\omega(1)$.
\end{proposition}

Now, we introduce a new type of network obtaining asymptotic truth learning under random orderings that relies primarily on constant-degree vertices. We contrast this with the celebrity graph, in which \textit{all} agents learn with high probability: there are no peripheral agents that do not learn yet still contribute significantly to the underlying learning scheme. In our construction, the constant-degree vertices do not learn themselves, but serve as guinea pigs for a select number of special vertices that propagate the signal to the rest of the network. Ignoring the constant-degree vertices, the rest of the network forms a complete graph, which is known not to achieve learning~\cite{Banerjee1992-ra}.

\begin{proposition}\label{prop:kcorecounterexample}
    There exists a network $\mathcal F$ obtaining random-order learning that no longer achieves learning when all vertices of degree $1$ are removed from each graph $G_n \in \mathcal{F}$. 
\end{proposition}

\begin{figure}
    \centering
    \begin{subfigure}[c]{0.42\linewidth}
        \centering
        \includegraphics[width=\linewidth]{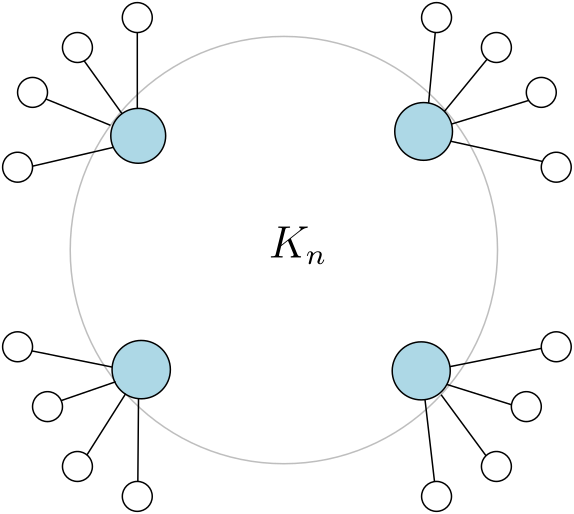}
        \caption{}
        \label{fig:covid}
    \end{subfigure}
    \hfill 
    \begin{subfigure}[c]{0.5\linewidth}
        \centering
        \includegraphics[width=\linewidth]{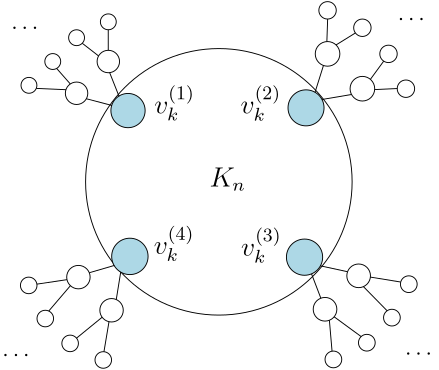}
        \caption{}
        \label{fig:embed}
    \end{subfigure}
    \caption{(a): By equipping a small number of vertices in a $K_n$ graph with many guinea pigs, the entire network can be boosted to achieve learning. (b): Example of $K_n$ being boosted by $\sqrt{n}$ complete binary trees of size $\log \log n$ to achieve random-order asymptotic learning.}
\end{figure}

An instance of such a network is shown in Figure \ref{fig:covid}. For any $g(n), h(n) = \omega(1)$ such that $g(n) \cdot h(n) = o(n)$, we can equip $g(n)$ special vertices in the complete graph $K_n$ with $h(n)$ guinea pigs with degree $1$. The number of added vertices is a negligible fraction of the entire network, and probability approaching $1$, each special vertex sees $\Omega(h(n))$ of its guinea pigs, which all give their own private signals. By a Chernoff bound, we may deduce that each special vertex learns, and so by \cref{cor:goodneighborhood}, the entire  graph $K_n$ achieves learning. 

Despite the complete graph having poor learning capabilities on its own, Proposition \ref{prop:kcorecounterexample} implies that it can achieve learning upon ``boosting'' a small number of its vertices.
One may note that equipping a vertex with a superconstant number of guinea pigs is not the only way to boost a vertex to achieve a high learning rate. Indeed, we show that \textit{any} construction permitting a vertex to achieve high learning rates under fixed orderings can be embedded into a $K_n$ network to boost the network at large.
One example of this using complete binary trees is shown in Figure \ref{fig:embed}.

Formally, suppose we are given a network $\mathcal{F'} = \{G_n'\}$ equipped with strategic orderings $\{\sigma_n\}$ such that for all $k \geq 1$, imposing the ordering $\sigma_k = \{v_1, v_2, \cdots, v_k\}$ on $G_k'$ permits $v_k$ to learn, i.e. $\ell_{\sigma_k}(v_k) = 1 - o(1)$. 
Now, for sufficiently large $n$, let $k = \log \log n$. We begin with the complete graph $K_{n}$, and pick $\sqrt{n}$ distinct vertices, denoted $S = \{v_{k}^{(i)}: i \in [\sqrt{n}]\}$. For each vertex $v_k^{(i)}$, attach a copy of $G'_{k} \setminus v_k$ so the induced graph on $\{G'_{k} \setminus v_k\} \cup v_k^{(i)} = G'_k$. 
For $m = n + \sqrt{n} (\log \log n - 1)$, we set $G_m$ to be the final result, and allow $\mathcal{F} = \{G_m\}$ to be our final network.

\begin{proposition}[Embedding arbitrary learning structures into random-order learning networks]\label{prop:embed}
Given any network $\mathcal{F}' = \{G_n'\}$ with strategic orderings $\{\sigma_n\}$ for which there exists $v = \{v_n: v_n \in G_n'\}$ obtaining asymptotic learning, the construction described above yields a network $\mathcal{F}$ that achieves random-order asymptotic learning.
\end{proposition}

A complete proof of both \cref{prop:kcorecounterexample} and \cref{prop:embed} can be found in \cref{appendix:results4}. To summarize the proof idea, the probability over the random ordering that the agents in any $G'_k$ copy are ordered identically to the sequence given by $\sigma_k$ is
$$\frac{1}{k!} \geq \frac{1}{(\log \log n)^{\log \log n}} \geq \frac{1}{n^{\delta}}$$
for any $\delta > 0$. For any copy of the graph $G'_k$ for which this occurs, the corresponding $v_k^{(i)}$ achieves learning by \cref{prop:improvement2}, i.e., its learning rate is at least that if it could only see its neighbors in the copy of $G'_k$. As the copies of $G_k'$ are disjoint, we conclude that each $v_k^{(i)}$ has a probability $\geq 1/n^{\delta}$ of achieving learning. There are $\sqrt{n}$ such copies of $v_k^{(i)}$, so we expect many copies to achieve learning.

The key technical challenge arises from the fact that the earlier copies of $v_k^{(i)}$ have an exponentially lower probability of learning, as necessarily all other vertices in its corresponding $G_k'$ copy must appear before it. To circumvent this, we condition on the number of graphs whose entire vertex set appears before an index $j = n/\omega(1)$ for some small superconstant denominator. Since the internal ordering of each graph copy is independent, we proceed as above to show that at least one vertex $v_k^{(i)}$ before index $j$ achieves learning. Thus, each vertex $u$ after index $j$ learns with probability approaching $1$, and as each agent's learning rate monotonically increases in its index, $u$ achieves learning in aggregate. 

We find this result rather surprising, as the inherent difficulty of preserving key graph properties under random orderings would suggest that complex structures such as binary trees cannot induce random-order learning. 
Notably, each boosted vertex $v_i$ has a low probability of learning on its own, but it is almost guaranteed that in aggregate, at least one $v_i$ learns. 

\subsection{Algorithmically Boosting Networks}\label{subsection:algorithm}

We now discuss how to enable networks to achieve random-order learning while minimizing the number of modifications (limited to edge insertions/deletions and vertex insertions),\footnote{
Our analysis excludes vertex deletions, as events on the induced subgraph of the original vertices may change non-uniformly under deletions.
} which can be thought of as a generalization of the boosting approach depicted in \cref{fig:covid}.
Our algorithm runs in polynomial time, taking a fixed graph as input. To obtain the asymptotic guarantees for an entire network, one would need to execute our algorithm on each $G_n \in \mathcal{F}$. The proofs in this section can be found in \cref{appendix:algproofs}.

A key assumption we make is that the algorithm has access to a learning oracle that can determine whether the learning rate of a vertex is above a certain threshold or not. 

\begin{definition}\label{def:oracle}
    Let $f(v, t): V \times (0,1)\to \{0, 1\}$ be an oracle that outputs $1$ if $v$ learns at rate $\geq t$ under random orderings, and $0$ otherwise. 
\end{definition}

To the best of our knowledge, it is unclear whether such a function can be computed efficiently. We discuss a useful heuristic to simulate a learning oracle in Section \ref{subsection:oracle}.

When modifying a graph, the number of new vertices should be sublinear in the number of vertices in the original graph so as not to overrule the existing network. 
However, the number of vertices requiring boosting can be as large as $\Omega(n)$, as is the case in a sparse graph with constant average degree. To resolve this imbalance, we replicate the structure of the celebrity graph by connecting each non-learner with a small but superconstant number of learning vertices. This subroutine is described in BoostAgents.

\begin{procedure}[H]
\caption{BoostAgents ($G$, $S$, $k$)}
\label{proc:boostagents}
\begin{algorithmic}  
    \STATE Construct $T = \{z_{11}, \dots, z_{kk}\} \cup \{w_1, \dots, w_k\}$
    \STATE Add edges $(z_{ij}, w_i)$ for all $i, j \in [k]$
    \STATE For each $v \in S$ and $i \in [k]$, add edge $(w_i, v)$
\end{algorithmic}
\end{procedure}

\begin{figure}
    \centering   
    \includegraphics[width=0.65\linewidth]{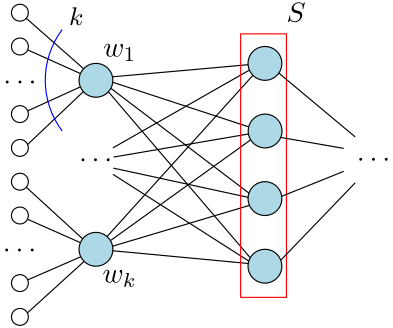}
    \caption{BoostAgents boosts the vertices $S \subseteq V$ by connecting them to an added set of celebrities $\{w_1, \cdots, w_k\}$.}
    \label{fig:boosting}
\end{figure}

As depicted in Figure~\ref{fig:boosting}, we introduce a set of $k$ ``celebrities" $\{w_i\}$, each equipped with $k$ independent guinea pigs to guarantee that their learning rate is high. By \cref{cor:goodneighborhood}, the added celebrities are capable of boosting the learning quality of the set $S$.

\begin{claim}
\label{claim:boost}
    Suppose $k = \omega(1)$. After BoostAgents on a set $S_n \subseteq V_n$, as $n \to \infty$, each $v \in S_n$ achieves random-order asymptotic truth learning.
\end{claim}

The number of vertices added by BoostAgents is $k^2 + k = \omega(1)$. As $k$ can be any arbitrary slow-growing function of $n$, the number of new vertices can be made to be within the optimum by an additive value that is slightly superconstant. However, the number of edges added is $|S|k + k^2$. While the additive term $k^2$ is small and can be ignored, $|S|$
may be as large as $\Omega(n)$. Hence, we focus on selecting a set $S$ that is of \textit{near-optimal} size, or within the optimum by a small multiplicative factor. 

To lower-bound the number of modifications necessary for a non-learning network to achieve learning, we note that the set of modified vertices must be able to reach almost all non-learners with constant probability.
It turns out that calling BoostAgents on a set of vertices maximizing coverage of non-learners is an appropriate method to enhance the learning rate of a network. We prove this in \cref{lemma:reachlearn}, where we argue that any non-learner reachable by many boosted vertices will also receive the effects of boosting.

\begin{lemma}\label{lemma:reachlearn}
    For a set $S \subseteq V$ on which BoostAgents has been called, suppose that each $u \in S$ has been equipped with $1/\varepsilon = \omega(1)$ many guinea pigs. Then any vertex $v$ that is reachable from $S$ w.p. $\geq 1 - \delta$ has learning rate at least $1 - \delta - 2\sqrt[3]{81\varepsilon}$.
    In particular, the probability $v$ is reachable by some $u \in S$ for which $\sigma(u) \geq n\varepsilon\alpha$ is at least $1 - \delta - \alpha\varepsilon$, for any $\alpha = \omega(1)$, $\alpha = o(\frac{1}{\varepsilon})$. 
\end{lemma}

The key challenge to proving \cref{lemma:reachlearn} is, again, monotonicity. We would like to boost vertices with high coverage, but coverage increases the earlier a vertex is in the ordering. A reasonable worry is that our algorithm will choose to boost vertices that only have high coverage when placed earlier than any of the celebrities added by BoostAgents, thereby avoiding the effects of boosting.
Fortunately, we show this is impossible in the second half of \cref{lemma:reachlearn}, which states that moving a vertex earlier by $\varepsilon n$ indices can only decrease the probability it reaches some other vertex by at most $\varepsilon$. Thus, having high coverage conditioned on being before all added celebrities implies high coverage more broadly.

Using \cref{lemma:reachlearn}, \cref{lemma:edgeapprox} follows immediately, suggesting that the minimum number of modifications needed through boosting alone, denoted $M_b(k, n)$, is within an arbitrarily slow-growing factor of the true optimum $M$.
As networks achieving asymptotic learning may leave $o(n)$ vertices uncovered, the exact optimum for our purposes may not be well defined. For precision, we define a tolerance parameter $T := T(n) = o(n)$ and mandate that in expectation over the random ordering, at most $T$ vertices are unaffected by the modifications. All our bounds hold for any threshold $T$.

\begin{lemma}
\label{lemma:edgeapprox}
    Given a network $\mathcal F = \{G_n\}_n$, consider any $M := M(n)$ such that it is possible for $\mathcal F$ to obtain random-order learning by performing at most $M(n)$ vertex insertions or edge insertions/deletions to each $G_n$, where in expectation at most $T := T(n)$ non-learning vertices are unaffected by the modifications. 
    Then there exists some $M_B := M_B(k,n)$ such that $M_B \leq O(k) M + 2k^2$ and it is possible to achieve random-order learning by adding at most $M_B(k,n)$ edges/vertices to each $G_n$ solely through calls to \textit{BoostAgents}$(G,S,k)$, with at most $T$ unaffected vertices in expectation.
\end{lemma}

It remains to develop an algorithm to find a set $S$ that maximizes its coverage. We leverage results from optimization problems with submodular functions~\cite{Nemhauser1978-jq,Wolsey1982-ik,Kempe03-maxinfluence}. Notably, fixing the size of the set $|S| = m$, the maximum coverage by the best $m$-element set can be approximated within a constant factor by simply greedily picking the agent with the largest marginal coverage increase at each iteration. Equivalently, fixing some target number of non-learners to be covered, 
the number of vertices to be boosted, chosen by the greedy approach is at most a $O(\log n)$ factor greater than the optimum. Our full algorithm is presented in \cref{alg:boostgraph-montecarlo}.

We first show that coverage as a function of a set $S$ is submodular. Let $V' \subseteq V$ be the set of non-learners in $G$, and $C_{\sigma}(S)$ denote the number of agents in $V'$ reachable from $S$ under ordering $\sigma$. Then, $C = \frac{1}{n!} \sum_{\sigma} C_{\sigma}(S)$ is the expected number of agents reachable from $S$ under a uniformly random ordering.

\begin{definition}\cite{Kempe03-maxinfluence}
    Given some domain space $U$, a submodular function $f: 2^U \to \mathbb{R}$ satisfies
    $$f(S \cup \{x\}) - f(S) \geq f(T \cup \{x\}) - f(T) $$    
    for all $S \subseteq T \subseteq U$ and $x \in U$.
\end{definition}

\begin{claim}\label{claim:submodular}
    For any $V' \subseteq V$, $C: 2^{V'} \to \mathbb{R}$ defined above is nonnegative, monotone, \& submodular.
\end{claim}

While computing $C(S)$ exactly is not computationally feasible, we may estimate it using Monte Carlo sampling to obtain $\tilde C(S) = \frac{1}{N}\sum_{i=1}^{N} C_{\sigma_i}(S)$, for orderings $\sigma_1, \dots, \sigma_N$ sampled independently and uniformly at random. For a sufficiently large $N$, $\tilde C(S)$ can be made within a negligible additive error of $C(S)$ at all times. 
Our approach is the following: starting with $S = \emptyset$, we iteratively compute the vertex $v$ maximizing $\tilde C(S \cup \{v\})$, adding it to $S$. We do this until the number of non-learners is less than $T$, at which point we boost $S$ and terminate. Using classic results from the submodular optimization literature, we may show that our randomized greedy algorithm obtains a $O(\log n)$-approximation.

\begin{algorithm}
\caption{BoostGraph-MonteCarlo$(k)$}
\label{alg:boostgraph-montecarlo}
\begin{algorithmic}[1]
    \STATE $S \gets \emptyset$, $\varepsilon = \frac{1}{n^2}$, $\delta = \frac{1}{n^{9}}$, $N = \frac{n^2}{2\varepsilon^2}\ln \frac{2}{\delta}$
    \STATE Using a learning oracle, determine the set of non-learners $V'$
    \WHILE{$\tilde C(S) < |V'| - T$}
        \STATE Sample orderings $\sigma_1, \dots, \sigma_N$ uniformly at random
        \STATE For each $v \in V'$, compute 
        $$\tilde C(S \cup \{v\}) = \frac{1}{N}\sum_{i=1}^{N} C_{\sigma_i}(S \cup \{v\})$$
        \STATE Let $v^* = \arg\max_{v \in V'} \tilde C(S \cup \{v\})$
        \STATE Set $S \gets S \cup \{v\}$
    \ENDWHILE
    \STATE Run $\textsc{BoostAgents}(G, S, k)$
\end{algorithmic}
\end{algorithm}

\begin{lemma}
\label{lemma:logapprox2}
    Let $ALG := ALG(k, n)$ be the number of edges added by Algorithm~\ref{alg:boostgraph-montecarlo}. Then for any $k \ll \log n$ and $M := M(n)$ as defined in \cref{lemma:edgeapprox}, $ALG \leq O(k \log n) \cdot M$ with probability at least $1 - \frac{1}{n^{8}}$.
\end{lemma}

The runtime of Algorithm ~\ref{alg:boostgraph-montecarlo} is polynomial, and by \cref{lemma:reachlearn} the modified network will obtain learning. Taking $k = g(n)$ for any $g(n) = \omega(1)$, we conclude with the following theorem:

\begin{theorem}\label{thm:algapx}
    Given efficient query access to a learning oracle, for any $k = \omega(1)$, $T = o(n)$, and network $\mathcal{F} = \{G_n\}$, \cref{alg:boostgraph-montecarlo} runs in polynomial time to modify each graph in $\mathcal{F}$ to achieve random-order learning such with probability $\geq 1 - \frac{1}{n^8}$, at most $T$ agents are unaffected in expectation.
    Furthermore, the number of modifications made $O(g(n)\log n)$-approximates the optimum, for any function $g(n) = \omega(1)$. 
\end{theorem}

\subsection{Efficiently Identifying Learning Vertices}\label{subsection:oracle}

We conjecture that it is NP-hard to compute, or even reasonably approximate, the output of any random-order learning oracle as defined in \cref{def:oracle}.
Even under a fixed ordering, computing the learning rate of a given vertex is difficult. In the repeated learning setting, in which agents have multiple attempts to predict the ground truth value, \cite{Hazla2021-vf} gives an exponential time algorithm to simulate the decision of any agent, and proves NP-hardness of computing any individual agent's decision when the private signals may be asymmetrically assigned (i.e. $\Pr(a_v = 1 \mid \theta = 0) \neq \Pr(a_v=1 \mid \theta = 1)$). Although their results are intended for the repeated learning setting, they are directly applicable to a graph coupled with a strategic ordering in the sequential setting. The random ordering adds an additional layer of complexity. 

\begin{conjecture}
    For general graphs $G$, $v \in V$, and $t \in [0, 1]$, computing the output of a learning oracle $f(v, t)$ is NP-hard.     
\end{conjecture}

We observe that one may instead identify simple local structures that permit learning for particular vertices.
Heuristically, many learning networks rely on agents that learn due to a specific graph structure independent of the rest of the network, which we refer to as ``first learners".
For random-order learning, one indicator of whether a vertex is a first learner with high probability over the ordering is the number of independent neighbors it has. For example, a vertex boosted by BoostAgents becomes a first learner as many of its neighbors will feed the targeted vertex their own independent private signals.

We describe a simple deterministic method to check if a vertex is a first learner. In particular, one may simply compute $\sum_{u \in N(v_n)} \frac{1}{\deg (u)}$ in $O(\deg (v)) = O(n)$ time. In essence, this tests whether an agent can directly see many independent signals from its neighbors. 

\begin{claim}\label{claim:srcneighbors}
    For $v = \{v_n: v_n \in G_n\}$, if  
    $$\sum_{u \in N(v)} \frac{1}{\deg (u)} = \omega(1)$$
    then $v$ learns under random orderings.
\end{claim}

The key insight is similar to the proof of \cref{prop:indep_set} in that for any $v \in G$, the expected number of neighbors arriving first within their own neighborhood under $\sigma$ is $\sum_{u \in N(v)} \frac{1}{1 + \deg (u)} =\Theta(\sum_{u \in N(v)} \frac{1}{\deg (u)})$.\footnote{
This value is closely related to the $\beta$-measure in directed social networks - see \cite{VanDenBrink2000}.
}
Each such vertex is forced to predict its own private signal. If $v$ can see the private signals of a superconstant number of neighbors, it will be able to learn with probability approaching $1$. A full proof can be found in \cref{appendix:results4}.

We observe that a graph may learn despite no agent having a high probability of being a first learner. The construction in \cref{prop:embed} shows that even arbitrarily complex structures with low probabilities of staying intact under random orderings can be made essential for a network's learning. Under strategic orderings, there exist networks capable of learning without a clear set of first learners, such as the butterfly graph described in \cite{Lu24-enabling}.
Unlike the heuristic in \cref{claim:srcneighbors}, which can be efficiently used to test whether a vertex is a first learner, it is unclear how one may feasibly determine if a network learns via either of these alternative mechanisms. 

\section{Conclusion} 

In this work, we explored the structure of truth learning in social networks under uniformly random decision orderings. Our results highlight the fragility of networks learning under strategic orderings, while simultaneously proving the robustness of random-order learning networks.
We also characterized various structural requirements, such as the necessity of large independent sets to achieve learning, and introduced constructions which demonstrate surprising capacity for learning under random orderings. 

Our algorithm for transforming arbitrary networks into random-order learning networks uses a greedy approach to pick agents in the network to boost.
With efficient access to an oracle that can calculate random-order learning rates for any vertex, this algorithm runs in polynomial time and succeeds with high probability. However \cite{Filip25-NPhard} proves that the fixed-order learning decision problem is NP-hard; we conjecture that a similar result holds for the random-order learning rate of a network, but speculate that an approach reducing from 3SAT as done in both \cite{Filip25-NPhard}, \cite{Cooper1990-fn} is unlikely to work due to their reliance on a strategic ordering.

Future work may explore the computational complexity of computing and/or approximating learning rates, and the design of alternative algorithms that do not rely on such oracles. One direction may be to derive stronger relationships between a network's learning rate and properties of its underlying graph topology.
To this end, it would be interesting to obtain stronger characterizations of learning networks, especially stronger sufficiency conditions on a network that would guarantee random-order asymptotic learning.

\section*{Acknowledgments}

This work was conducted when Guo and Xiong were participants of the 2025 DIMACS REU program (funded by CCF-2447342) at Rutgers University. 
In addition, Guo was also supported by NSF 2421503 supplement to DMS-2220271 and Xiong was supported by AI Institute (ACTION) under IIS-2229876. Gao would like to acknowledge NSF support through IIS-2229876, DMS-2220271, CNS-2515159, DMS-2311064, CCF-2208663 and CCF-2118953.


\bibliographystyle{ACM-Reference-Format}
\bibliography{learning}

@ARTICLE{Wolsey1982-ik,
  title     = "An analysis of the greedy algorithm for the submodular set
               covering problem",
  author    = "Wolsey, L A",
  journal   = "Combinatorica",
  publisher = "Springer Nature",
  volume    =  2,
  number    =  4,
  pages     = "385--393",
  month     =  dec,
  year      =  1982,
  language  = "en"
}

@ARTICLE{Nemhauser1978-jq,
  title     = "An analysis of approximations for maximizing submodular set
               functions—I",
  author    = "Nemhauser, G L and Wolsey, L A and Fisher, M L",
  journal   = "Math. Program.",
  publisher = "Springer Science and Business Media LLC",
  volume    =  14,
  number    =  1,
  pages     = "265--294",
  month     =  dec,
  year      =  1978,
  language  = "en"
}

@article{Cooper1990-fn,
  title = "The computational complexity of probabilistic inference using
           bayesian belief networks",
  author = "Cooper, Gregory F",
  journal = "Artif. Intell.",
  publisher = "Elsevier BV",
  volume = 42,
  number = "2-3",
  pages = "393--405",
  month = mar,
  year = 1990,
  language = "en",
}

@article{Hazla2021-vf,
  title = "Bayesian decision making in groups is Hard",
  author = "Hązła, Jan and Jadbabaie, Ali and Mossel, Elchanan and Rahimian, M
            Amin",
  journal = "Oper. Res.",
  publisher = "Institute for Operations Research and the Management Sciences
               (INFORMS)",
  volume = 69,
  number = 2,
  pages = "632--654",
  month = mar,
  year = 2021,
  language = "en",
}

@inproceedings{Domingos2001-min,
    author = {Domingos, Pedro and Richardson, Matt},
    title = {Mining the network value of customers},
    year = {2001},
    isbn = {158113391X},
    publisher = {Association for Computing Machinery},
    address = {New York, NY, USA},
    url = {https://doi.org/10.1145/502512.502525},
    doi = {10.1145/502512.502525},
    booktitle = {Proceedings of the Seventh ACM SIGKDD International Conference on Knowledge Discovery and Data Mining},
    pages = {57–66},
    numpages = {10},
    location = {San Francisco, California},
    series = {KDD '01}
}

@inproceedings{Borgs2014-maxinfl,
    author = {Borgs, Christian and Brautbar, Michael and Chayes, Jennifer and Lucier, Brendan},
    title = {Maximizing social influence in nearly optimal time},
    year = {2014},
    isbn = {9781611973389},
    publisher = {Society for Industrial and Applied Mathematics},
    address = {USA},
    booktitle = {Proceedings of the Twenty-Fifth Annual ACM-SIAM Symposium on Discrete Algorithms},
    pages = {946–957},
    numpages = {12},
    location = {Portland, Oregon},
    series = {SODA '14}
}

@ARTICLE{Hann-Caruthers2018-ec,
  title     = "The speed of sequential asymptotic learning",
  author    = "Hann-Caruthers, Wade and Martynov, Vadim V and Tamuz, Omer",
  journal   = "J. Econ. Theory",
  publisher = "Elsevier BV",
  volume    =  173,
  pages     = "383--409",
  month     =  jan,
  year      =  2018,
  language  = "en"
}

@article{Mossel2017-sd,
  title = "Opinion exchange dynamics",
  author = "Mossel, Elchanan and Tamuz, Omer",
  journal = "Probab. Surv.",
  publisher = "Institute of Mathematical Statistics",
  volume = 14,
  number = "none",
  pages = "155--204",
  month = jan,
  year = 2017,
}

@article{Acemoglu2011-tx,
  title = "Opinion dynamics and learning in social networks",
  author = "Acemoglu, Daron and Ozdaglar, Asuman",
  journal = "Dyn. Games Appl.",
  publisher = "Springer Science and Business Media LLC",
  volume = 1,
  number = 1,
  pages = "3--49",
  month = mar,
  year = 2011,
  language = "en",
}

@ARTICLE{Mossel2014-eb,
  title     = "Asymptotic learning on Bayesian social networks",
  author    = "Mossel, Elchanan and Sly, Allan and Tamuz, Omer",
  journal   = "Probab. Theory Relat. Fields",
  publisher = "Springer Science and Business Media LLC",
  volume    =  158,
  number    = "1-2",
  pages     = "127--157",
  month     =  feb,
  year      =  2014,
  language  = "en"
}

@inproceedings{Lu24-enabling,
  author = "Kevin Lu and Jordan Chong and Matt Lu and Jie Gao",
  title = "Enabling Asymptotic Truth Learning in a Social Network",
  booktitle = "Proceedings of the 20th Conference on Web and Internet Economics
               (WINE'24)",
  pages="530–547",
  month = "December",
  year = "2024",
}

@inproceedings{Filip25-NPhard,
  author = "Filip {\'U}radn{\'i}k and Amanda Wang and Jie Gao",
  title = "Maximizing Truth Learning in a Social Network is {NP-hard}",
  booktitle = "Proceedings of the 24th International Conference on Autonomous Agents and Multiagent Systems
(AAMAS 2025)",
    pages="2078-2086",
  month = "February",
  year = "2025",
}

@book{Easley2010-networks,
  title = {Networks, crowds, and markets: Reasoning about a highly connected
           world},
  author = {Easley, David and Kleinberg, Jon},
  volume = {1},
  year = {2010},
  publisher = {Cambridge university press},
  url = {https://www.cs.cornell.edu/home/kleinber/networks-book/},
}

@inproceedings{Hazla2019-reasoning,
  title = {Reasoning in Bayesian opinion exchange networks is {PSPACE}-hard},
  author = {H\k{a}z\l{}a, Jakub and Jadbabaie, Ali and Mossel, Elchanan and
            Rahimian, Mohammad Ali},
  booktitle = {Conference on Learning Theory},
  pages = {1614--1648},
  year = {2019},
  organization = {PMLR},
  url = {https://proceedings.mlr.press/v99/hazla19a.html},
}

@inproceedings{Kempe03-maxinfluence,
  title = {Maximizing the spread of influence through a social network},
  author = {Kempe, David and Kleinberg, Jon and Tardos, Eva},
  booktitle = {Conference on Learning Theory},
  pages = {137--146},
  year = {2003},
  organization = {ACM},
  url = {https://doi.org/10.1145/956750.956769},
}

@article{Mossel_2013,
  title = {Majority dynamics and aggregation of information in social networks},
  volume = {28},
  ISSN = {1573-7454},
  url = {http://dx.doi.org/10.1007/s10458-013-9230-4},
  DOI = {10.1007/s10458-013-9230-4},
  number = {3},
  journal = {Autonomous Agents and Multi-Agent Systems},
  publisher = {Springer Science and Business Media LLC},
  author = {Mossel, Elchanan and Neeman, Joe and Tamuz, Omer},
  year = {2013},
  month = jun,
  pages = {408–429},
}

@book{Chamley2004-or,
  title = "Rational Herds: Economic Models of Social Learning",
  author = "Chamley, Christophe",
  publisher = "Cambridge University Press",
  address = "Cambridge",
  year = 2004,
  language = "en",
}

@article{Welch1992-yt,
  title = "Sequential sales, learning, and cascades",
  author = "Welch, Ivo",
  journal = "J. Finance",
  publisher = "Wiley",
  volume = 47,
  number = 2,
  pages = "695--732",
  month = jun,
  year = 1992,
  copyright = "http://onlinelibrary.wiley.com/termsAndConditions\#vor",
  language = "en",
}

@article{VanDenBrink2000,
    title = "Measuring domination in directed networks",
    author = "Ren{\' e} van den Brink and Robert P. Gilles",
    journal = "Social Networks",
    volume = 22,
    number = 2,
    year = 2000,
    pages= "141--157",
}

@article{Bikhchandani1992-rs,
  title = "A Theory of Fads, Fashion, Custom, and Cultural Change as
           Informational Cascades",
  author = "Bikhchandani, Sushil and Hirshleifer, David and Welch, Ivo",
  journal = "J. Polit. Econ.",
  publisher = "The University of Chicago Press",
  volume = 100,
  number = 5,
  pages = "992--1026",
  month = oct,
  year = 1992,
}

@article{Banerjee1992-ra,
  title = "A simple model of herd behavior",
  author = "Banerjee, A V",
  journal = "Q. J. Econ.",
  publisher = "Oxford University Press (OUP)",
  volume = 107,
  number = 3,
  pages = "797--817",
  month = aug,
  year = 1992,
  language = "en",
}

@article{Sgroi2002-rz,
  title = "Optimizing Information in the Herd: Guinea Pigs, Profits, and Welfare
           ",
  author = "Sgroi, Daniel",
  journal = "Games Econ. Behav.",
  volume = 39,
  number = 1,
  pages = "137--166",
  month = apr,
  year = 2002,
}

@phdthesis{Smith1991-sy,
  title = "Essays on dynamic models of equilibrium and learning",
  author = "Smith, Lones",
  year = 1991,
  school = "University of Chicago",
}

@inproceedings{Arieli2021-social,
  author = {Arieli, Itai and Sandomirskiy, Fedor and Smorodinsky, Rann},
  title = {On Social Networks That Support Learning},
  year = {2021},
  isbn = {9781450385541},
  booktitle = {Proceedings of the 22nd ACM Conference on Economics and
               Computation},
  pages = {95–96},
  numpages = {2},
  location = {Budapest, Hungary},
  series = {EC '21},
}

@ARTICLE{Mobius2014-oy,
  title     = "Social Learning in Economics",
  author    = "Mobius, Markus and Rosenblat, Tanya",
  journal   = "Annu. Rev. Econom.",
  publisher = "Annual Reviews",
  volume    =  6,
  number    =  1,
  pages     = "827--847",
  month     =  aug,
  year      =  2014
}

@article{Bahar2020-am,
  title = "Multi-issue social learning",
  author = "Bahar, Gal and Arieli, Itai and Smorodinsky, Rann and Tennenholtz,
            Moshe",
  journal = "Math. Soc. Sci.",
  volume = 104,
  pages = "29--39",
  month = mar,
  year = 2020,
}

@article{Acemoglu2011-vj,
  title = "Bayesian Learning in Social Networks",
  author = "Acemoglu, Daron and Dahleh, Munther A and Lobel, Ilan and Ozdaglar,
            Asuman",
  journal = "Rev. Econ. Stud.",
  publisher = "Oxford Academic",
  volume = 78,
  number = 4,
  pages = "1201--1236",
  month = mar,
  year = 2011,
  language = "en",
}

@article{Aislinn_Bohren2016-ya,
  title = "Informational Herding with Model Misspecification",
  author = "Aislinn Bohren, J",
  journal = "J. Econ. Theory",
  publisher = "SSRN",
  volume = 163,
  pages = "222--247",
  year = 2016,
  language = "en",
}

@unpublished{Mueller-Frank2018-lh,
  title = "Manipulating Opinions in Social Networks",
  author = "Mueller-Frank, Manuel",
  month = oct,
  year = 2018,
}

@inproceedings{Hann-Caruthers2025-sq,
author = {Hann-Caruthers, Wade and Pan, Minghao and Tamuz, Omer},
title = {Network and Timing Effects in Social Learning},
year = {2025},
isbn = {9798400719431},
publisher = {Association for Computing Machinery},
address = {New York, NY, USA},
url = {https://doi.org/10.1145/3736252.3742486},
booktitle = {Proceedings of the 26th ACM Conference on Economics and Computation},
pages = {6},
numpages = {1}
}

@article{Smith2000-wk,
  title = "Pathological Outcomes of Observational Learning",
  author = "Smith, Lones and S{\o}rensen, Peter",
  journal = "Econometrica",
  publisher = "John Wiley \& Sons, Ltd",
  volume = 68,
  number = 2,
  pages = "371--398",
  month = mar,
  year = 2000,
}

@incollection{Golub2017-qo,
  author = {Golub, Benjamin and Sadler, Evan},
  isbn = {9780199948277},
  title = "Learning in social networks",
  booktitle = "The Oxford Handbook of the Economics of Networks",
  publisher = {Oxford University Press},
  address = "Oxford",
  year = {2016},
  month = {04},
  Pages = "504–542",
}

@techreport{caro1979new,
  title={New results on the independence number},
  author={Caro, Yair},
  year={1979},
  institution={Technical Report, Tel-Aviv University}
}

@misc{wei1981lower,
  title={A lower bound on the stability number of a simple graph},
  author={Wei, Victor K},
  year={1981},
  publisher={Bell Laboratories Technical Memorandum Murray Hill, NJ, USA}
}

\appendix
\section{Proofs of Helpful Results}\label{appendix:background}

\begin{proof}[Proof of Lemma \ref{lemma:conditionallearning}]
    We can decompose the learning rate of $v$ as
    \[
    \ell(v) = \ell(v \mid A) \Pr(A) + \ell(v \mid  \overline A) \Pr(\overline A).
    \]
    Trivially we have $\ell(v \mid \overline A) \leq 1$, and $\Pr(\overline A) = 1 - \Pr(A)$ so we have
    \[
    \ell(v \mid A) \Pr(A) + 1 - \Pr(A) \geq \ell(v) \geq 1 - \varepsilon.
    \]
    This simplifies to
    \[
    \ell(v \mid A) \geq 1 - \frac{\varepsilon}{\Pr(A)}
    \]
    as desired.
\end{proof}

To prove the general improvement principle (\cref{prop:improvement2}), we leverage a simpler claim: that each agent performs as least as well as its best neighbor.
\begin{proposition}[Proposition $1$ in \cite{Lu24-enabling}]
\label{prop:improvement1}
    Given a graph/ordering pair $G, \sigma$ and a particular vertex $v \in V$, $\ell_{\sigma}(v) \geq \max\{\ell_{\sigma}(u) : u \in N(v), \sigma(u) < \sigma(v)\}$.
\end{proposition}

\begin{proof}[Proof of \cref{prop:improvement2}]



    In $G$, construct $v'$ to be a copy of $v$ along with its incident edges  in $G'$, so $N(v') \subseteq N(v)$. Then, construct an ordering $\sigma'$ on $G \cup \{v'\}$ by inserting $v'$ immediately before $v$ in $\sigma$. For sake of analysis, we now tweak the setting such that the private signals $p_{v'}, p_v$ are forced to be identical instead of independent. (This is different than conditioning on the event $p_{v'} = p_v$, which would give information to both agents. One may think of this modification as $v, v'$ collectively receiving a single private signal.) As each $v, v'$ still has private signal strength $q$, the subnetworks observed by $v'$/$v$ are identical to those seen under $\sigma$, and $v'$ doesn't affect the decision of $v$, we have that $\ell_{\sigma}(v) = \ell_{\sigma'}(v)$ and $\ell_{\sigma}'(v) = \ell_{\sigma'}(v')$. 

    Now, consider adding an edge between $v'$ and $v$. As both $v, v'$ are Bayesian and all information available to $v'$ is also available to $v$, $v$ does not change its decision due to the addition of this edge, so its learning rate remains unchanged. But then by \cref{prop:improvement1}, $\ell_{\sigma'}(v) \geq \ell_{\sigma'}(v')$. Thus, $\ell_{\sigma}(v) \geq \ell_{\sigma}'(v)$. 
\end{proof}

We may intuit that any agent with many learning neighbors itself must also learn. However, this is not a straightforward conclusion. Conditioned on arriving before the agent itself, the learning rate of any given neighbor is necessarily lowered. This is particularly harmful even when many of the neighbors arrive beforehand. 
We prove a lower bound on the learning rate of any vertex with many learning neighbors, incurring a $\log d$ multiplicative factor in the error of the agent.

\begin{lemma}
\label{lemma:goodneighborhood}
    Any vertex $v$ with at least $d$ neighbors obtaining learning rates at least $1 - \varepsilon$ will itself achieve a learning rate of at least
    \[
    \ell(v) \geq 1-  \frac{1}{d+1} - 2\varepsilon \ln d.
    \]
\end{lemma}

While it may seem counterintuitive that this bound weakens with the degree of the agent for large $d$, by \cref{prop:improvement2}, we may substitute any $d' \leq d$ into the above bound. 

\begin{proof} 
    Let $M$ be the set of neighbors $u \in N(v)$ such that $\ell(u) \geq 1 - \varepsilon$. For any ordering $\sigma$ of the vertices, let $M_\sigma \subseteq M$ contain all $u \in M$ such that $u$ arrives before $v$. If there are $d$ neighbors achieving learning $1 - \varepsilon$, the probability $v$ arrives in a certain position relative to the $d$ neighbors is
    \[
    \Pr(\abs{M_\sigma} = k) = \frac{1}{d + 1}
    \]
    for all $k$ such that $0 \leq k \leq d$. Now, we condition on the event $\abs {M_{\sigma}} = k$. If $k = 0$ then $\ell(v \mid M_\sigma = \emptyset) = q$ trivially.
    Otherwise, for a fixed vertex $u_0 \in M$, consider the event $A_{u_0,k}$ where $u_0$ is the $k$th vertex to arrive among all vertices in $M \cup \{ v \}$ and $v$ arrives after $u_0$. Also define $B_{u_0, k}$ as the event that one of $A_{u_0, 1}, A_{u_0, 2}, \dots, A_{u_0, k}$ occurs. We have
    \[
    \Pr(A_{u_0,k}) = \frac{d + 1 - k}{d(d+1)},
    \]
    and since all $A_{u_0, k}$ are disjoint for a fixed $u_0$, this gives
    \[
    \Pr(B_{u_0, k}) = \sum_{i=1}^k \Pr(A_{u_0,k}) = \frac{kd - k(k-1)/2}{d(d+1)}.
    \]
    By \cref{lemma:monotonicity}, we also observe that for all $k' < k$, we have
    \[
    \ell(u_0 \mid A_{u_0, k}) \geq \ell(u_0 \mid A_{u_0, k'}).
    \]
    Thus we have $\ell(u_0 \mid A_{u_0,k}) \geq \ell(u_0 \mid B_{u_0, k})$. By \cref{lemma:conditionallearning}, we have
    \[
    \ell(u_0 \mid A_{u_0, k}) \geq \ell(u_0 \mid B_{u_0, k}) \geq 1 - \frac{\varepsilon}{\Pr(B_{u_0,k})}.
    \]
    Now we relate this learning rate to the learning rate of $v$: define the event $C_{u_0, k}$ where $u_0$ is the $k$th vertex to arrive and $v$ is the $k + 1$th vertex to arrive among all vertices in $M \cup \{ v \}$.
    Note that $C_{u_0, k}$ is a subset of $A_{u_0, k}$ and
    \[
    \ell(u_0 \mid C_{u_0, k}) = \ell(u_0 \mid A_{u_0, k})
    \]
    since the distribution of subnetworks of $u_0$ is indistinguishable across events $C_{u_0, k}$ and $A_{u_0, k}$. By
    \cref{prop:improvement1},
    \[
    \ell(v \mid C_{u_0, k}) \geq \ell(u_0 \mid C_{u_0, k}).
    \]
    For fixed values of $k$ and averaging over all $u_0 \in N(v)$, the $C_{u_0, k}$ are all equally likely, disjoint, and are exhaustive in the event $\abs{M_\sigma} = k$, so we have
    \[
    \ell(v \mid \abs{M_\sigma} = k) = \frac{1}{d}\sum_{u_0 \in N(v)} \ell(v \mid C_{u_0, k}) \geq 1 - \frac{d(d+1)}{kd - k(k-1)/2}\varepsilon
    \]
    Finally, we average over all possible values of $\abs{M_\sigma}$ to get
    \begin{align*}
        \ell(v) &\geq \frac{1}{d+1} \sum_{k=0}^d \ell(v \mid \abs{M_\sigma} = k) \\
        &\geq \frac{1}{d+1} \sum_{k=1}^d \ell(v \mid \abs{M_\sigma} = k) \\
        &\geq \frac{d}{d+1} - \frac{\varepsilon}{d+1} \sum_{k=1}^d \frac{d(d+1)}{kd-k(k-1)/2} \\
        &=\frac{d}{d+1}-\varepsilon d \sum_{k=1}^d \frac{1}{kd-k(k-1)/2} \\
        &\geq \frac{d}{d+1} - \varepsilon d \sum_{k=1}^d \frac{2}{kd} \\
        &\geq 1 - \frac{1}{d+1} - 2\varepsilon\ln d
    \end{align*}
    as desired.
\end{proof}

\begin{proof}[Proof of Corollary~\ref{cor:goodneighborhood}]
    Let $\{v_n\}_n$ be the sequence of vertices $v_n \in G_n$ represented by $v$. There exists some $\varepsilon = \varepsilon(n)$ and some sequence $\{d_n\}_n$ such that $v_n$ has at least $d_n$ neighbors with learning rate at least $\varepsilon(n)$, and $d_n \to \infty$ and $\varepsilon \to 0$ as $n \to \infty$.

    For simplicity, we weaken the bound in \cref{lemma:goodneighborhood} for $v_n$ as
    \[
    \ell(v_n) \geq 1 - \frac{1}{d} - 2 \varepsilon \ln d.
    \]
    This bound works for any value of $d \leq d_n$. Take $d = \min(d_n, e^{1/\sqrt{\varepsilon(n)}})$. Thus we have
    \[
    \ell(v_n) \geq 1 - \frac{1}{d_n} - 4 \sqrt{\varepsilon}.
    \]
    Both $\frac{1}{d_n} \to 0$ and $\sqrt{\varepsilon} \to 0$ as $n \to \infty$, so
    \[
    \ell(v_n) \to 1 
    \]
    and thus $v$ learns.
\end{proof}

\begin{proof}[Proof of Lemma \ref{lemma:monotonicity}]
    Fix any ordering $\sigma$ for which $\sigma(v) = i$. Let $u: \sigma(u) = i + 1$ be the vertex at index $i + 1$. Denote $\sigma'$ to be the ordering nearly identical to $\sigma$ with a single inversion between indices $i$ and $i + 1$ that swaps $u$ and $v$. We note this establishes a bijection between the orderings in $A_i$ and $A_{i+1}$. 
    
    If $u \notin N(v)$, $v$ sees the same neighbors in both orderings. As the neighbors' behaviors are unaffected, $v$ must behave identically under both $\sigma$ and $\sigma'$, so $\Pr(a_v = \theta | \sigma) = \Pr(a_v = \theta | \sigma')$. Otherwise, if $u \in N(v)$, then $v$ sees the prediction $a_u$ under $\sigma'$ not previously visible in $\sigma$. By Proposition~\ref{prop:improvement2}, we have that $\Pr(a_v = \theta | \sigma) \leq \Pr(a_v = \theta | \sigma')$. In either case, we sum over all $\sigma$ and corresponding $\sigma'$ to get that $\Pr(a_v = \theta | A_i) \leq \Pr(a_v = \theta | A_{i+1})$, as desired.
\end{proof}

\section{Fragile Strategic-Order Networks}\label{appendix:fragile}

We find that fixed-order networks are surprisingly fragile to changes in various parameters. Prior work, such as \cite{Lu24-enabling} and \cite{Arieli2021-social}, analyzes networks which achieve learning for all values of $q$, the parameter which controls the accuracy of private signals. However, here we demonstrate two networks which have significantly different behaviors at different values of $q$. One learns only when $q$ is high, and one learns only when $q$ is low. 


\begin{proof}[Proof of \cref{prop:changeq}]
    Consider network $\mathcal F = \{ G_n \}_{n \geq 1}$ illustrated in \cref{fig:smarter-dumber}. 
    $G_n$ contains the vertices $u_0, v, v_1, v_2, v_3, v_4$, the vertices $w_i$ and $w_i'$ for $1 \leq i \leq \log n$, and the vertices $u_j$ for $1 \leq j \leq n - 2\log n - 6$. There are directed edges $(v_i, v)$ for each $v_i$, edges $(v, w_i)$ and $(w_i', w_i)$ for each $w_i$, and edges $(u_{i-1}, u_i)$ for each $u_i$ where $i \geq 1$.

    Then, observe that $v_i$ and $w_i'$ are sources in the directed graph. Since $u_0$ has a learning rate strictly greater than $q$, each $u_i$ will ignore its own private signal \& follow the signal of $u_0$ passed down from the agent $u_{i-1}$ before it. Thus, each $u_i$ has the same learning rate as $u_0$. As $n\to \infty$, the set $\{u_i\}$ consists of almost all the vertices in the network, implying that $L(G_n) \to \ell(u_0)$ as $n \to \infty$.

    We have that agent $v$ learns the correct answer with probability
    \[
        \ell(v, q) = q^5 + 5q^4 (1 - q) + 10q^3(1 - q^2).
    \]
    Recall $a_v$ denotes the prediction that agent $v$ makes, and $p_{w_i}$ denotes the private signal of $w_i$. Let $S_i = \{p_{w_i}, a_{w_i'}, a_v\}$ be the set of relevant information available to agent $w_i$. 
    Since its corresponding $w_i'$ always transmits its own private signal, each $w_i$ essentially obtains two private signals. If the signals differ, or both are the same as $v$, then $w_i$ will predict the same value as $v$. Otherwise if both signals are opposite $v$'s signal, then $w_i$ infers that the probability $v$ is correct is
    \[
        P(a_v = \theta | S_i) = \frac{f(q)}{f(q) + f(1-  q)},
    \]
    where $f(q) = (1-q)^2 \ell(v, q)$.
    There exists a value $q_0 \approx 0.7887$ such that if $q < q_0$, then $P(\theta = a_v) < \frac{1}{2}$ so $w_i$ follows its private signal which is the opposite of $a_v$, and if $q > q_0$ then $P(\theta = a_v) > \frac{1}{2}$ so $w_i$ predicts $a_v$.

    Suppose $q < q_0$. If $a_v = \theta$, then $a_{w_i} = \theta$ so long as $w_i$ and $w_i'$ aren't both given the incorrect private signal, occurring with independent probability $1 - (1 - q)^2$. Otherwise if $a_v = 1 - \theta$, then $a_{w_i} = \theta$ with probability $q^2$. Note that $1 - (1 - q)^2 > q^2$ for all $0 < q < 1$. By the weak law of large numbers, as $n \to \infty$, the fraction of the $a_{w_i}$ which are correct approaches either $q^2$ or $1 - (1-q)^2$ with high probability. Depending on which occurs, $u_0$ can directly determine the correct ground truth, so $\ell(u_0) \to 1$ and thus $\mathcal F$ achieves asymptotic learning.

    If $q > q_0$ then $a_{w_i} = a_v$ always, so $u_0$ only observes the value of $a_v$. Since $\ell(v) > q$, $u_0$ follows the observed signal and achieves learning rate $\ell(u) = \ell(v)$. This is constant for arbitrarily large values of $n$ so $\mathcal F$ does not achieve asymptotic learning.
\end{proof}

We use the same idea to construct a network that does the opposite, achieving asymptotic learning for high values of $q$ and not for low values of $q$.

\begin{proposition}
There exists a fixed-order network $\mathcal F$ that obtains asymptotic learning for $q > q_0$ and doesn't learn for $q < q_0$, for some fixed threshold $q_0$.
\end{proposition}

\begin{figure}
    \centering
    \begin{tikzpicture}
        \begin{scope}[every node/.style={circle, draw}]
            \node (v1) at (0, 0) {};
            \node (v2) at (1, 0) {};
            \node (v3) at (2, 0) {};
            \node (u0) at (2.5, 3.5) {$u_0$};
            \foreach \i in {1,...,3} {
                \node (w\i) at (4.5, \i + 1) {$w_\i$};
                \node (ww\i) at (6, \i + 1) {$v_{\i}$}; 
                \path [->] (ww\i) edge (w\i);
                \path [->] (v1) edge (w\i);
                \path [->] (v2) edge (w\i);
                \path [->] (v3) edge (w\i);
                \path [->] (w\i) edge (u0);
            }
        \end{scope}
        \node (vdots) at (4.5,5) {$\vdots$};
        \node (dots) at (1, 3.5) {$\cdots$};
        \path [->] (u0) edge (dots);
    \end{tikzpicture}
    \caption{Network that obtains learning for $q > q_0$ but not for $q < q_0$. For each $i$, we have condensed the vertices $v_{i0}, \dots, v_{i4}$ into one gadget vertex $v_i$ for visual clarity.}
    \label{fig:dumber-smarter}
\end{figure}
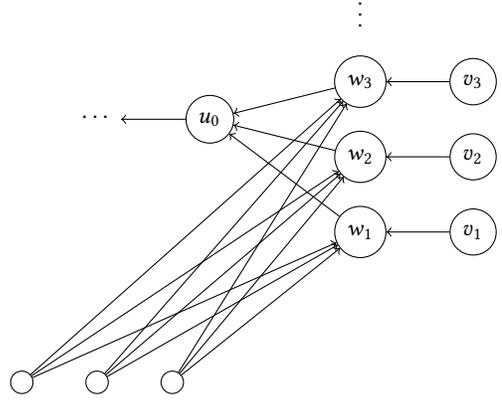

\begin{proof}

Consider the network $\mathcal F$ shown in \cref{fig:dumber-smarter}, defined as $\mathcal F = \{ G'_n \}_{n \in \mathbb N}$ as follows:
$G'_n$ contains the vertices $x_1,x_2,x_3$, the vertices $w_i$ for $1 \leq i \leq \log n$, the gadgets $v_{i}$ for $1 \leq i \leq \log n$ consisting of vertices $v_{ij}$ for $0 \leq j \leq 4$, and the vertices $u_i$ for $0 \leq i \leq n - 5\log n - 4$. $G'_n$ contains the edges $(x_1, w_i)$, $(x_2, w_i)$, $(x_3, w_i)$, $(v_{i0}, w_i)$ and $(v_{ij}, v_{i0})$ for all $1 \leq i \leq \log n$ and $1 \leq j \leq 4$. It also contains the edges $(w_i, u_0)$ for all $1 \leq i \leq \log n$ and $(u_{i - 1}, u_i)$ for all $1 \leq i \leq n - 6\log n - 4$. As in the previous graph, note that $L(G'_n) \to \ell(u_0)$ as $n \to \infty$. 

We again have that $\ell(v, q) = q^5 + 5q^4(1-q) + 10q^3(1-q^2)$. 
Let $S_i = \{p_{w_i}, a_{v_i0}, a_{x_1}, a_{x_2}, a_{x_3}\}$ denote the information available to node $w_i$. As $x_1, x_2, x_3$ each transmit their own private signal, each $w_i$ functionally gets $4$ private signals and the action of $v_{i0}$. 

Consider when $q > q_0$ (using the same value of $q_0 \approx 0.7887$). For any $i \in [\log n]$, if the $4$ private signals received by $w_i$ are split evenly (two $0$s and two $1$s) or there is a strict majority in favor of $a_{v_{i0}}$, $w_i$ will follow the prediction of $v_{i0}$. Otherwise, if there are three predictions for $1 - a_{v_{i0}}$, we have that
$$\Pr(a_{v_{i0}} = \theta | S_i) = \frac{g(q)}{g(q) + g(1-q)} > \frac{1}{2}$$
for $g(q) = q(1-q)^3\ell(v, q)$. 

We now case on the predictions of $x_1, x_2, x_3$. If all $3$ are given the correct private signal, $w_i = \theta$ so long as either $p_{w_i} = \theta$ or $a_{v_{i0}} = \theta$, which occurs with probability $q(1 - \ell(v_{i0})) + \ell(v_{i0}) \in (q, \ell(v_{i0}))$. If exactly $1$ or $2$ are correct, $a_{w_i} = a_{v_{i0}}$ unconditionally so $w_i = \theta$ with probability $\ell(v_{i0})$. If all are wrong, $w_i = \theta$ only when $p_{w_i} = \theta$ and $a_{v_{i0}} = \theta$, which occurs with probability $q\ell(v_{i0})$. For any $q > q_0$, these probabilities are all distinct and strictly greater than $1/2$, so as $n \to \infty$, the proportion of $w_i$ predicting the majority vote accounts for $q^* \in (1/2, 1)$ of the $w_i$'s, where $q^* \in \{q(1 - \ell(v_{i0})) + \ell(v_{i0}), \ell(v_{i0}), q\ell(v_{i0})\}$. By the weak law of large numbers, $u_0$ will be able to determine the ground truth with high probability by simply taking the majority of the predictions among the $w_i$, in which case $\mathcal{F}$ achieves asymptotic learning.

Otherwise, when $q < q_0$, with probability $(1-q)^3$, all $x_1, x_2, x_3$ predict $1 - \theta$. In this case, each $w_i$ follows the signals of these agents regardless of the prediction of $v_{i0}$. Subsequently, $u_0$ and the rest of the network follow this signal, which is incorrect with constant probability. Thus, $\mathcal{F}$ does not achieve asymptotic learning. 
\end{proof}

\section{Correctness of Algorithm \ref{alg:boostgraph-montecarlo}}\label{appendix:algproofs}

\begin{proof}[Proof of \cref{claim:boost}]   
    By uniformity, any celebrity $w_i$ sees at least $\sqrt{k}$ of its guinea pigs with probability $\geq 1 - 1/\sqrt{k}$. Conditioned on this event, $w_i$ sees at least $1 + \sqrt{k}$ private signals. As per Chpt. 16 of \cite{Easley2010-networks}, $w_i$ predicts the correct ground truth so long as at least half of the signals are correct, which occurs with probability $\geq 1 - \exp(-(1+\sqrt{k})/(8q^2)) = 1 - \exp(-\Theta(k))$ by a standard Chernoff bound. Then by \cref{prop:improvement2}, each $w_i$ learns with probability $\geq 1 - O(1/\sqrt{k})$.
    As $k = \omega(1)$, all $w_i$'s learn. Each $v \in S_n$ contains the vertices $w_1, \ldots w_k$ as neighbors, so the result follows immediately by Corollary~\ref{cor:goodneighborhood}.
\end{proof}

We now provide a simple combinatorial result that will be helpful in many of the following proofs, namely \cref{lemma:reachlearn}:

\begin{lemma}\label{lemma:chebyshev}
    For a given value $k \in [n]$, suppose set $S = \{v_1, \dots, v_k\}$ is sampled uniformly at random across all $k$-subsets of $[n]$. Let $Y$ be the random variable denoting the number of elements of $S$ with index at most $n\varepsilon$. Then 
    $$\Pr(k\varepsilon / 2 \leq Y  \leq 3k\varepsilon/2) \geq 1 - \frac{4}{k\varepsilon}$$
\end{lemma}

\begin{proof}
    Let $X_i$ be an indicator random variable equal to $1$ iff $v_i$ arrives before index $n\varepsilon$, and observe that $Y = \sum_{i = 1}^k X_i$. By linearity, 
    $$\mathbb{E}[Y] = \sum_{i=1}^{k} \mathbb{E}[X_i] = k\varepsilon$$
    where we use that $\mathbb{E}[X_i] = \Pr(X_i = 1) = \varepsilon$ for all $i \in [k]$. We now compute 
    $$\text{Var}(Y) = \sum_{i = 1}^{k} \text{Var}(X_i) + \sum_{i \neq j} \text{Cov}(X_i, X_j)$$
    We note $\text{Var}(X_i) \leq \mathbb{E}[X_i^2]  = \mathbb{E}[X_i]$ as $X_i \in \{0, 1\}$, and thus $\sum_{i=1}^{k} \text{Var}(X_i) \leq \mathbb{E}[Y]$. Furthermore, 
    \begin{align*}
        \text{Cov}(X_i, X_j) &= \mathbb{E}[X_i X_j] - \mathbb{E}[X_i]\mathbb{E}[X_j] \\
        &= \frac{n\varepsilon(n\varepsilon - 1)}{n(n-1)} - \varepsilon^2 \\
        &\leq 0
    \end{align*}
    and thus $\text{Var}(Y) \leq \mathbb{E}[Y]$. By Chebyshev's inequality,
    $$\Pr(Y < k\varepsilon / 2 \cup Y > 3k\varepsilon/2) \leq \Pr(|Y - \mathbb{E}[Y]| \geq k\varepsilon / 2) \leq \frac{\text{Var}(Y)}{k^2\varepsilon^2/4} \leq \frac{4}{k\varepsilon}$$
    which gives us $\Pr(k\varepsilon / 2 \leq Y \leq 3k\varepsilon / 2) \geq 1 - \frac{4}{k\varepsilon}$. 
\end{proof}

\begin{proof}[Proof of \cref{lemma:reachlearn}]
    Let $E$ be the event that $v$ is reachable from the set of vertices in $S$, and $A$ be the event $\sigma(v) \leq n - \varepsilon \alpha n$ for some unspecified $\alpha = \omega(1), \alpha = o(\frac{1}{\varepsilon})$. We are given that $\Pr(E) = 1 - \delta$, and also have that $\Pr(A) = 1 - \alpha \varepsilon$. Then,
    $$\Pr(E) = \Pr(E \cap A) + \Pr(E \cap A^c) \leq \Pr(E \cap A) + \alpha\varepsilon $$
    $$\therefore \Pr(E \cap A) \geq 1 - \delta - \alpha \varepsilon$$

    We now compute the probability $v$ is reached by some vertex in $S$ after index $n\varepsilon \alpha$.
    Let $\sigma[x:y]$ denote the subgraph induced by vertices between indices $x, y \in [n]$ in the ordering, and let $E'$ be the event that $v$ is reachable by some $u \in S$ where $\sigma(u) \geq n\varepsilon\alpha$. We claim $\Pr(E') \geq \Pr(E \cap A)$. Indeed, for any ordering $\sigma \in E \cap A$, observe we may bijectively map each $\sigma$ to a new ordering $\sigma' = \sigma[n  - n\varepsilon\alpha + 1 : n]\circ [1: n - n\varepsilon \alpha]$, i.e. we move the suffix of length $n\varepsilon\alpha$ to the front of the ordering. Note that $v$ is reachable by some $u \in \sigma[1: n - n\varepsilon\alpha]$, and in $\sigma'$, $\sigma'(u) \geq n\varepsilon\alpha$. Thus, $\sigma' \in E'$, and as this holds for all $\sigma \in E \cap A$, we conclude $\Pr(E') \geq 1 - \delta - \alpha\varepsilon$. 

    By \cref{prop:improvement1}, we have that $\ell(v | E') \geq \ell(u | E')$, where $u \in S, \sigma(u) \geq n\varepsilon\alpha$ is a vertex in $v$'s subnetwork. 
    It remains for us to prove $\ell(u | E')$ is large by showing that, conditioned on $\sigma(u) \geq n\varepsilon \alpha$, at least one of the celebrities $w_i$ added in BoostAgents arrives before $u$ and still learns. 
    
    For each $i \in [1/\varepsilon]$, consider the celebrity $w_i$, and let $B_i$ be the event that $w_i$ appears before index $n\varepsilon \alpha$. Let $X = \sum_{i} B_i$ denote the total number of celebrities before index $n\varepsilon \alpha$. By \cref{lemma:chebyshev}, $X \in [\frac{\alpha}{2}, \frac{3\alpha}{2}]$ with probability at least $1 - \frac{4}{\alpha}$. We note conditioning on event $E'$ only increases the probability of $B_i$, as $E'$ only concerns indices after $n\varepsilon \alpha$. 

    For any $w_i$ appearing before index $n \varepsilon\alpha$, let $C_i$ denote the event that it sees at least $\sqrt{\alpha}$ of its guinea pigs. By \cref{prop:improvement2}, the learning rate of $w_i$ conditioned on $C_i$ is at least as good as
    \[
    \ell(w_i | C_i) \geq (1 - \exp(-\frac{\sqrt{\alpha}}{8q^2})),
    \]
    using our previously described Chernoff bound. Thus it suffices to show $\Pr(C_i)$ is large. Again applying \cref{lemma:chebyshev}, we have probability at least $1 - \frac{4}{\alpha}$ that at least $\alpha / 2$ of $w_i$'s guinea pigs
    appear before index $n\varepsilon \alpha$. Taking a uniformly random ordering, all orderings of $w_i$ with relation to the other guinea pigs before index $n\varepsilon \alpha$ have equal probability, so $w_i$ has probability at least $1 - \frac{2}{\sqrt{\alpha}}$ of seeing at least $\sqrt{\alpha}$ guinea pigs. Thus, $\Pr(C_i | B_i) \geq 1 - \frac{4}{\alpha} - \frac{2}{\sqrt{\alpha}} \geq 1 - \frac{4}{\sqrt{\alpha}}$ for sufficiently large $\alpha$.

    We now have that 
    \begin{align*}
        \ell(u | E') &\geq \ell(u | E', \exists \; i: B_i = 1) \Pr(\exists \; i: B_i = 1 | E') \\
        &\geq \ell(w_i | E', B_i = 1)\left( 1 - \frac{4}{\alpha} \right)\\
        &\geq \ell(w_i | B_i = 1, C_i) \Pr(C_i  | B_i) \left(1 - \frac{4}{\alpha} \right) \\
        &\geq \left(1  - e^{-\sqrt{\alpha} / (8q^2)} \right)\left(1 - \frac{4}{\sqrt{\alpha}}\right)\left(1 - \frac{4}\alpha\right) \\
        &\geq 1 - \frac{9}{\sqrt{\alpha}}
    \end{align*}
    for $\alpha = \omega(1)$. Finally, we have that 
    $$\ell(v) \geq \ell(v | E') \Pr(E') \geq (1 - 9/\sqrt{\alpha})(1 - \delta - \alpha\varepsilon) \geq 1 - \delta - 2\sqrt[3]{81\varepsilon}$$
    where our desired inequality is obtained by setting $\alpha = \sqrt[3]{\frac{81}{\varepsilon^{2}}}$.    
\end{proof}

\begin{proof}[Proof of Lemma \ref{lemma:edgeapprox}]
    Consider any modification strategy using at most $M := M(n)$ modifications to $G := G_n$, $G = \{V, E\}$, where the resulting modified graph $G_{M}$ obtains learning. Let $S \subseteq V$ be the set of all vertices $u$ that are involved in a modification (i.e., an edge incident on $u$ is added/deleted). We show that calling BoostAgents$(G, S, k)$ on $G$ to obtain the modified graph $G_B$ results in a network that learns, for any superconstant $k = \omega(1)$. 

    Consider any vertex $v$ not originally a learner under $G$ that learns in $G_{M}$. With some probability, $v$'s subnetwork must contain some vertex $u \in S$ involved in a modification. Let $A$ be the event $v$ is reachable by any such $u$, and denote $\ell^{M}(v)$ to be the random-order learning rate of $v$ in $G_{M}$. Then,
    $$\ell^{M}(v) = \ell^{M}(v | A) \Pr(A) + \ell^{M}(v | A^c)\Pr(A^c) = 1 - o(1)$$
    We now show the learning rate of $v$ in $G_B$, $\ell(v)$, is at most a $o(1)$ additive correction less than $\ell^{M}(v)$. 
    In both $G_{M}$ and $G_B$, conditioned on $A^c$, the subnetwork of $v$ is identical to that in $G$, and thus the learning rate remains unchanged. Furthermore, $\Pr(A^c)$ depends only on the relative ordering of the original vertices in $V$ and is hence identical in both graphs. 
    Now, let $A'$ denote the probability that there exists a modified vertex $u \in S$ within $v$'s subnetwork appearing after index $n\varepsilon\alpha$, where we may set $\alpha = 1/\sqrt{\varepsilon}$. By the analysis of \cref{lemma:reachlearn}, $\Pr(A') \geq \Pr(A) - \sqrt{\varepsilon}$, and $\ell(v | A') \geq 1 - 5\varepsilon^{1/4} \geq \ell^{M}(v) - o(1)$. Thus,
    \begin{align*}
        \ell(v) &= \ell(v | A) \Pr(A) + \ell(v | A^c) \Pr(A^c) \\
        &\geq \ell(v | A') \Pr(A') + \ell(v | A^c)\Pr(A^c) \qquad (A' \subseteq A)\\
        &\geq (\ell^{M}(v | A) - o(1))(\Pr(A) - o(1)) + \ell^{M}(v | A^c)\Pr(A^c) \\
        &\geq \ell^{M}(v) - o(1)
    \end{align*}
    and thus $v$ learns in $G_B$. As this holds for all viable non-learning $v \in G$ and since $G_{M}$ learns, we conclude that $G_B$ learns as well.
    
    Let $ALG$ be the number of edges/vertices added to construct $G_B$. Note that BoostAgents adds at most $2k$ edges for each edge modification in $G_{M}$, along with $k^2$ additional edges and $k^2 + k$ vertices. This gives us $ALG \leq k \cdot (2M+1) + 2k^2 = O(k) \cdot M + 2k^2$, as desired.
\end{proof}

\begin{proof}[Proof of \cref{claim:submodular}]
    Nonnegativity and monotonicity are inherent to reachability: that is, $C_{\sigma}(T) \geq C_{\sigma}(S) \geq 0$ for any $T \supseteq S$. On submodularity, we fix $\sigma$ and consider some agent $v$. For any sets $S \subseteq T \subseteq V'$, we have that $C_{\sigma}(S) \subseteq C_{\sigma}(T)$. As such, $C_{\sigma}(\{v\}) \setminus C_{\sigma}(S) \supseteq C_{\sigma}(\{v\}) \setminus C_{\sigma}(T)$. But $\frac{1}{n!} \sum_{\sigma}|C_{\sigma}(\{v\}) \setminus C_{\sigma}(S)| = f(S \cup \{v\}) - f(S)$ and $\frac{1}{n!} \sum_{\sigma}|C_{\sigma}(\{v\}) \setminus C_{\sigma}(T)| = f(T \cup \{v\}) - f(T)$, so we are done.
\end{proof}

\begin{proof}[Proof of Lemma \ref{lemma:logapprox2}]
    At any iteration, let $\tilde c_v = \tilde C(S \cup \{v\}) - \tilde C(S)$. 
    By Monte Carlo guarantees, our estimate $\tilde C(S \cup \{v\})$ is within an additive $\varepsilon$ of the true value $C (S \cup \{v\})$ with probability $\geq 1 - \delta = 1 - \frac{1}{n^9}$. Across all $n$ iterations, the probability that no estimate falls outside of the error is $1 - \frac{1}{n^8}$. 

    Even so, our algorithm may pick a suboptimal $v$ at each iteration. 
    Let $S^{(i)}$ be the set of size $i$ that a exponential time greedy algorithm with exact knowledge of $C(S)$ would obtain, and let $S_{ALG}^{(i)}$ be the set of size $i$ obtained by our algorithm (equivalently, this is the set obtained after $i$ iterations of the main loop). As we incur $\leq \varepsilon = \frac{1}{n^2}$ error at each iteration, we have for all $i \leq n$, 
    $$C(S^{(i)}) \leq C(S_{ALG}^{(i)}) + i\varepsilon \leq C(S_{ALG}^{(i)}) + \frac{1}{n}$$
    Let $S$ be the final set obtained by the exponential time greedy algorithm, covering $\geq |V'| - T$ vertices, and $S_{ALG}$ be our algorithm's final boosting set. Suppose $|S| = j$, i.e. after $j$ iterations, the exponential time greedy algorithm has $C(S^{(j)}) \geq |V'| - T$.     
    Since each uncovered vertex $v$ has probability $\frac{1}{1 + \deg v} \geq \frac{1}{2n}$ of being first among its neighbors \& unreachable from $S$, each iteration increments $\tilde{C}(S^{(i)}_{ALG})$ by at least $\frac{1}{2n}$. 
    We conclude that $C(S_{ALG}^{j+2}) \geq |V'| - T$ with high probability.     
    Hence, the final set $S_{ALG}$ obtained by our randomized algorithm is at most $2$ larger than the set $S$ obtained if we could compute $C$ exactly. 
    For the rest of our analysis, we now bound the size of $S$, knowing $S_{ALG} \leq (1 + \frac{2}{|S|})S$.

    Let $S^*$ be the optimal boosted set. It remains for us to show that $|S| \leq |S^*| \cdot O(\log n)$, which directly implies $ALG \leq M_B \cdot O(\log n)$ for any $M_B$. This would be sufficient to obtain $ALG \leq M \cdot O(k \log n) + 2k^2$ as the boosting construction using $M_B$ modifications in \cref{lemma:edgeapprox} adds at most $2k$ edges for each of the $M$ modifications made under an arbitrary strategy. For sufficiently small $k = O(\log n)$, the additive $O(k^2)$ term can be merged with the $M \cdot O(k \log n)$ term, giving us the desired result.
    
    Let $u_i = |V'| - C(S)$ denote the number of non-learning vertices still uncovered by \cref{alg:boostgraph-montecarlo} when $|S| = i$. The exponential time algorithm terminates when $u_i \leq T$. Note that $|V'| - C(|S^*|) \leq T$; consequently, at any point in the algorithm, there exists a sequence of at most $M_B$ vertices that we may add to $S$ to cover all but at most $T$ non-learners, as $|V'| - C(S \cup S^*) \leq T$.    
    By submodularity of coverage (\cref{claim:submodular}), we may invoke the greedy approximation guarantee in \cite{Kempe03-maxinfluence} to obtain
    $$u_{i+M_B} \leq \frac{u_i}{e}$$
    for any iteration $i \geq 0$. As $u_0 = |V'|$, after $i$ iterations,
    $$u_{i}  \leq |V'| \left(\frac{1}{e} \right)^{\lfloor \frac{i}{M_B} \rfloor}$$
    The algorithm terminates when $u_i < T$, so the number of iterations of the main loop is $i \leq M_B \cdot \ln \frac{|V'|}{T} \leq M_B \cdot O(\log n)$. The number of iterations is precisely the size of $S$, and thus $|S| \leq O(\log n) |S^*|$.    
\end{proof}

\begin{proof}[Proof of Theorem \ref{thm:algapx}]
    By \cref{lemma:logapprox2}, \cref{alg:boostgraph-montecarlo} $O(k\log n)$-approximates the number of modifications needed to achieve learning with high probability. Taking $k = g(n)$ for any $g(n) = \omega(1)$, $g(n) = O(\log n)$, we obtain the desired approximation guarantee. 
    Furthermore, \cref{alg:boostgraph-montecarlo} outputs a set $S$ covering at least $|V'| - T$ vertices w.p. $\geq 1 - \frac{1}{n^8}$. 
    By \cref{lemma:reachlearn}, with probability approaching $1$, all but $o(n)$ vertices are boosted to learn, so the algorithm output is indeed a learning network. 
    
    It remains for us to analyze the runtime of Algorithm~\ref{alg:boostgraph-montecarlo}. Sampling $N$ random orderings takes $O(n^7 \log n)$, assuming each ordering can be generated in $O(n)$ time. For each ordering, computing $\tilde c_v = \tilde C(S \cup \{v\}) - \tilde C(S)$ can be done with a single BFS per ordering, which takes $O((m + n) n^6 \log n)$ time per vertex and thus $O((m+n) n^7 \log n)$ time in aggregate. The runtime of BoostAgents is $O(kn) = O(n^2)$, so the overall runtime is $O((m+n) n^7 \log n)$, which is polynomial.
\end{proof}

\section{Additional Proofs in Section 4}\label{appendix:results4}

\begin{proof}[Proof of Proposition \ref{prop:indep_set}]
    We will prove the contrapositive. 
    Let $\alpha(G)$ denote the maximum independent set size of $G$, and let the set of agents guaranteed to predict their private signal be $S \subseteq V$. We note that any agent $v \in V \setminus S$ will follow its neighbors if they unanimously agree on some binary value.
    However, the only agents guaranteed to predict their private signals are sources (in-degree $0$), and potentially some agents with in-degree $1$ who predict their private signal due to the tie-breaking rule. For each vertex $v$, this occurs only if $v$ is one of the first two vertices in the ordering among $\{v \} \cup N(v)$. 
    A classical result (first shown in \cite{caro1979new} and \cite{wei1981lower}) states that
    \[
        \alpha(G) \geq \sum_{v \in V} \frac{1}{1 + \deg(v)},
    \]
    so the expected number of agents with in-degree $0$ or $1$ is at most $\sum_{v \in V} \frac{2}{1 + \deg(v)} \leq 2 \alpha(G)$. Thus with probability at least $1/2$ over the random ordering, $|S| \leq 4\alpha(G)$ by Markov's inequality. When $\alpha(G) = O(1)$, we observe the probability of all agents in $S$ being given the wrong signal (and hence predicting the wrong binary value) is at least $(1 - q)^{4\alpha(G)} = \Omega(1)$. 
    Should this occur, we inductively observe that all vertices in $V \setminus S$ will herd to the wrong binary value. 
    Thus, with fixed probability, the entire graph herds to the incorrect binary value and therefore doesn't learn.
\end{proof}

\begin{proof}[Proof of \ref{prop:kcorecounterexample}]
    For any sufficiently large $n$, our construction leverages any functions $h(n), g(n) = \omega(1)$ for which $h(n) \cdot g(n) = o(n)$.

    We first take the complete graph $K_n$ and select $g(n)$ special vertices, denoted by the set $S$. Equip each special vertex with $h(n)$ guinea pigs (vertices with degree $1$). Let $G_m$ be the resulting graph, for $m = n + h(n)g(n)$. Without the vertices of degree $1$, we are left with $K_n$, which doesn't achieve asymptotic learning under any ordering, as discussed in \cite{Bikhchandani1992-rs}.

    To show the network $\mathcal{F} = \{G_m\}$ learns, it suffices to show that the vertices within the $K_n$ subgraph learn, as $o(m)$ vertices fall outside this subgraph.
    For $u \in S$, let $A$ be the event that $u$ sees at least $\sqrt{h(n)}$ many of its guinea pigs. Over the uniformly random ordering, $\Pr(A) \geq 1 - \frac{1}{\sqrt{h(n)}}$. As long as at least half of its guinea pigs predict the correct ground truth, so too will $u$ (see Chpt. 16 of \cite{Easley2010-networks}).
    By a Chernoff bound and Proposition \ref{prop:improvement2}, conditioned on seeing at least $\sqrt{h(n)}$ many independent private signals, we have that
    $$\ell(u | A) \geq 1 - \exp\left(-\frac{\sqrt{h(n)}}{8q^2}\right)$$
    so by Lemma \ref{lemma:conditionallearning}, we have $$\ell(u) \geq \left(1 - \frac{1}{\sqrt{h(n)}}\right)\left(1 - e^{-\sqrt{h(n)}/(8q^2)}\right).$$ For $h(n) = \omega(1)$, we conclude that each special vertex learns. 
    
    Since all vertices in the $K_n$ subgraph have at least $g(n) = \omega(1)$ learning neighbors, by Corollary \ref {cor:goodneighborhood}, all such vertices achieve learning as well. We conclude that $\mathcal{F}$ achieves asymptotic learning.
\end{proof}

\begin{proof}[Proof of Proposition \ref{prop:embed}]
    Within any graph $G_m$ (constructed as per the proposition statement), let $V^{(i)}$ be the vertex set of the $i$th copy of the graph $G_k'$, and $v_{k}^{(i)} \in V^{(i)}$ be the corresponding vertex in $V^{(i)} \cap S$. Recall $k = \log \log n$. We fix an index $j = n/k$, and for each $i \in \sqrt{n}$, let $A_i$ be the indicator random variable equal to $1$ iff all vertices in $V^{(i)}$ appear at/before index $j$. By uniformity over the ordering, we compute
    \begin{align*}
        \Pr(A_i = 1) &= \frac{\prod_{l=0}^{k - 1} (j - l)}{\prod_{l=0}^{k -1} (n-l)} = \prod_{l = 0}^{k - 1} \frac{j - l}{m - l} \\
        &\geq \left(\frac{j - k}{m - k}\right)^{k} \\
        &\geq \left( \frac{1}{2\log\log n} \right)^{\log \log n} \\
        &\geq \frac{1}{\log n (\log \log n)^{\log \log n}} \\
        &\geq \frac{1}{n^{\delta}}
    \end{align*}
    for any $\delta > 0$. Let $X$ be a random variable denoting the number of graphs $G_k'$ that appear at/before index $j$. As $X = \sum_{i \in [\sqrt{n}]} A_i$, by linearity, $\mathbb{E}[X] \geq n^{1/2 - \delta}$. Furthermore, by linearity of variance, $$\mathrm{Var}(X) = \sum_{i \in [\sqrt{n}]} \mathrm{Var}(A_i) + \sum_{i \neq l} \mathrm{Cov}(A_i, A_l)$$
    We first note $\mathrm{Var}(A_i) \leq \mathbb{E}[A_i^2] = \mathbb{E}[A_i]$ as $A_i \in \{0, 1\}$. We also observe $\mathrm{Cov}(A_i, A_l) \leq 0$, as 
    \begin{align*}
        &\mathrm{Cov}(A_i, A_l) = \mathbb{E}[A_i A_l] - \mathbb{E}[A_i]\mathbb{E}[A_l] \\
        =& \frac{\prod_{l=0}^{2k - 1} (j - l)}{\prod_{l=0}^{2k-1} (m - l)} - \left(\frac{\prod_{l=0}^{k - 1} (j - l)}{\prod_{l=0}^{ k -1} (m - l)} \right)^2 \\
        \leq &\left(\frac{\prod_{l=0}^{k - 1} (j - l)}{\prod_{l=0}^{ k -1} (m - l)} \right)^2 - \left(\frac{\prod_{l=0}^{k - 1} (j - l)}{\prod_{l=0}^{ k -1} (m - l)} \right)^2 \\
        = & 0
    \end{align*}
    where the penultimate line uses that
    $$\frac{\prod_{l=k}^{2k - 1} (j - l)}{\prod_{l=k}^{ 2k -1} (m - l)} \leq \frac{\prod_{l=0}^{k - 1} (j - l)}{\prod_{l=0}^{ 2k-1} (m - l)}$$
    so $\mathrm{Var}(X) \leq \mathbb{E}[X]$. Hence, by Chebyshev's inequality,
    $$\Pr(|X - \mathbb{E}[X]| \geq \mathbb{E}[X]/2) \leq \frac{\mathbb{E}[X]}{\mathbb{E}[X]^2/4} \leq \frac{4}{n^{1/2 - \delta}}$$
    so we conclude that with probability $\geq 1 - \frac{4}{n^{1/2 - \delta}}$ that at least $n^{1/2 - \delta}/2$ sets $V^{(i)}$ appear entirely before index $j$. 

    Without loss of generality, let $S' = [1, \dots, L]$ be the indices corresponding to sets appearing before index $j$. Denote $B_i$ to be the event that the relative ordering on $V^{(i)}$ is identical to $\sigma_k$. As the vertex sets $V^{(i)}$ are disjoint, the events $B_i$ for $i \in S'$ are mutually independent with each other and with the events $A_i$. Thus,
    \begin{align*}
        \Pr\left(\bigcap_{i \in S'} \overline{B_i} \bigm\vert |S'| \geq n^{1/2 - \delta} /2\right) &= \prod_{i \in S'} \Pr\left(\overline{B_i} \bigm\vert |S'| \geq n^{1/2 - \delta} \right) \\
        &= \prod_{i \in S'} \Pr(\overline{B_i}| A_i) = \prod_{i \in S'} \Pr(\overline{B_i})\\
        &= (1 - 1/k!)^{L}
    \end{align*}
    where the last line uses that $\Pr(B_i) = 1/k!$. When $L \geq n^{1/2 - \delta} / 2$, since $1/k! \leq 1/(\log \log n)^{\log \log n} \leq 1/n^{\delta}$, we use that $1 - x \leq e^{-x}$ to get that $(1 - 1/k!)^L \leq \exp(-n^{1/2 - 2\delta}/2)$. Thus,
    $$\Pr\left(\bigcup_{i \in S'} B_i \bigm\vert S' \right) = 1 - \Pr \left( \bigcap_{i \in S'} \overline {B_i} \bigm\vert S' \right) \geq 1 - \exp(-n^{1/2 - 2\delta}/2)$$

    Finally, to show our network learns, consider any vertex $u$ in the subgraph $u \in K_{n}$. Let $\ell_{\sigma_k}(v_k) = 1 - \varepsilon = 1 - o(1)$, and note that this lower bounds $\ell(v_k^{(i)} | A_i = 1)$ by \cref{prop:improvement2}. Let $\ell(u | j)$ be the learning rate of $u$ conditioned on being at index $j$. By \cref{lemma:monotonicity}, $\ell(u) \geq \ell(u | j)\Pr(\sigma(u) \geq j) = \ell(u | j)(1 - \frac{1}{\log \log n})$, it remains to show $\ell(u | j)$ goes to $1$ as $n \to \infty$. 
    
    With probability at least $\Pr(|S'| \geq n^{1/2 - \delta}, \bigcup_{i \in S'} B_i)$, any $u$ at vertex $j$ sees at least one $v_{k}^{(i)}$ that learns with probability $1 - \varepsilon$. Thus, 
    $$\ell(u | j) \geq (1 - \varepsilon)\left(1 -  \exp\left(\frac{-n^{1/2 - 2\delta}}{2}\right)\right)\left(1 - \frac{4}{n^{1/2 - \delta}}\right) = 1 - o(1)$$
    for any $0 < \delta < 1/4$. We conclude that each $u \in K_{n}$ learns, and as the $K_n$ subgraph consists of all but $o(m)$ vertices, the network $\mathcal{F} = \{G_m\}$ obtains asymptotic random-order learning. 
\end{proof}

\begin{proof}[Proof of \cref{claim:srcneighbors}]
    Fix $n$ and let $v := v_n$. Let $m = \sum_{u \in N(v)} \frac{1}{\deg u} = \omega(1)$, $N'(v) =\{u \in N(v): \deg u \leq \deg v / \sqrt{m}\}$, and $d = |N'(v)|$. Note that $d > m$, and
    \begin{align*}
        \sum_{u \in N(v_n)} \frac{1}{\deg u} &= \sum_{u \in N'(v)} \frac{1}{\deg u} + \sum_{u \notin N'(v)} \frac{1}{\deg u} \\
        &\leq \sum_{u \in N'(v)} \frac{1}{\deg u} + \frac{\deg v}{\deg v / \sqrt{m}} \\
        &= \sum_{u \in N'(v)} \frac{1}{\deg u} + \sqrt{m}
    \end{align*}
    Thus, $\sum_{u \in N'(v)} \frac{1}{\deg u}  \geq m - \sqrt{m} \geq m/2$. 
    
    Observe that a superconstant number of vertices in $N'(v)$ appear very early in the ordering. In particular, let $X$ be the number of vertices in $N'(v)$ arriving before index $\frac{nl}{d}$. By \cref{lemma:chebyshev}, $l/2 \leq X \leq 3l/2$ with probability at least $1 - \frac{4}{l}$. 
    
    Let $S \subseteq N'(v)$ be the set of vertices appearing before index $\frac{nl}{d}$. We now show that with high probability, all vertices in $S$ are the first among their neighbors. 
    For each $u \in S$, let $B_u$ be the event that $u$ is the first among its neighbors. For each $j \in [\deg u]$, let $B_{uj}$ be the event that neighbor $j$ is after index $\frac{nl}{d}$. Then $\bigcap_{j} B_{uj} \cap \{u \in S\} \subseteq B_{u} \cap \{u \in S\}$, and
    $$\Pr\left(\bigcup_{j} \overline{B_{uj}} \bigm\vert u \in S \right) \leq \frac{l \deg u}{d} \leq \frac{l}{\sqrt{m}}$$
    by a union bound, where $\Pr(\overline{B_{uj}} | u \in S) \leq \frac{l}{d}$. Thus, $\Pr(B_u | u \in S) \geq 1 - \Pr(\bigcup_{j} \overline{B_{uj}} | u \in S) \geq 1 - l/\sqrt{m}$. We conclude that $\Pr(\overline{B_u} | u \in S) \leq l/\sqrt{m}$, and thus
    \begin{align*}
        \Pr\left(|S| \in \left[\frac{l}{2}, \frac{3l}{2}\right], \bigcap_{u \in S} B_u \right) &= \Pr\left(\bigcap_{u \in S} B_u \bigm\vert S \right) \Pr\left(|S| \in \left[\frac{l}{2}, \frac{3l}{2}\right]\right)\\
        &\geq \left(1 - \Pr\left(\bigcup_{u \in S} \overline{B_u} \bigm\vert S\right) \right)\left(1 - \frac{4}{l}\right) \\
        &\geq \left(1 - \frac{3l^2}{2\sqrt{m}}\right)\left(1  -\frac{4}{l} \right)
    \end{align*}
    where the last line uses a union bound to obtain $\Pr(\bigcup_{u \in S} \overline{B_u} | S) \leq \frac{3l}{2}\Pr(\overline{B_u} | u \in S)$. We may set $l = m^{1/6} = \omega(1)$ to get that, with probability $1 - \Theta(m^{-1/6}) = 1 - o(1)$, $|S| \in [l/2, 3l/2]$ and each $u \in S$ is the first among its neighbors. Let the intersection of both these events be denoted $E$. 

    To conclude, we show the learning rate of our vertex $v$ conditioned on being at index $nl/d$ is high. As $\ell(v | i) \geq \ell(v | nl/d)$ for any $i \geq nl/d$, by Lemma \ref{lemma:monotonicity}, we note $\ell(v) \geq \ell(v | \frac{nl}{d}) \Pr(\sigma(v) > \frac{nl}{d}) \geq \ell(v | \frac{nl}{d})(1 - \frac{l}{d})$. As $1 - \frac{l}{d} = 1 - o(1)$, it remains to show $\ell(v | \frac{nl}{d})$ goes to $1$ as $n \to \infty$. We have that
    \begin{align*}
    \ell\left(v \bigm\vert \frac{nl}{d}\right) \geq & \ell\left(v \bigm\vert \frac{nl}{d}, E\right) \Pr\left(E \bigm\vert \sigma(v) = \frac{nl}{d}\right) \\
    \geq &\ell\left( v \bigm\vert \frac{nl}{d}, E\right)(1 - \Theta(m^{-1/6}))
    \end{align*}
    Conditioned on $E$ and being after index $nl/d$, vertex $v$ sees at least $m^{1/6}/2$ neighbors who arrive first in their neighborhood, who simply predict their private signal. By a Chernoff bound, the learning rate obtained by taking the majority of these private signals is $\geq 1 - \exp(-m^{1/6}/(16q^2))$, so by Proposition \ref{prop:improvement2}, $\ell(v | \frac{nl}{d}, E) \geq 1 - \exp(-\Theta(m^{1/6}))$. We conclude that $v$ learns.
\end{proof}




\end{document}